\documentclass{llncs}
\usepackage[utf8]{inputenc}
\usepackage{enumerate}
\usepackage{qip}
\usepackage{framed}
\usepackage{tikz}
\usepackage[section]{placeins}
\usepackage[T1]{fontenc} % for \textsc in section title

\usetikzlibrary{fit,calc}

\newcommand{\siml}{{\cal S}}
\newcommand{\env}{\mathcal Z}

\newcommand{\ot}{\textsf{OT}}
\newcommand{\bc}{\textsf{BC}}
\newcommand{\onecc}{\textsf{1CC}}
\newcommand{\twocc}{\textsf{2CC}}

\newcommand{\imaxacc}{I_{\max}^{\mathrm{acc}}}
\newcommand{\imax}{I_{\max}}

\renewcommand{\succ}{P_{\mathrm{succ}}}
\newcommand{\adapt}{P^{\mathrm{A}}_{\mathrm{succ}}}
\newcommand{\nonadapt}{P^{\mathrm{NA}}_{\mathrm{succ}}}
\newcommand{\guess}{P_{\mathrm{guess}}}

\def\Ball{B^{\delta}}
\def\u01{U_{0,1}}

\begin{document}

\title{Adaptive Versus Non-Adaptive Strategies in the Quantum Setting
  with Applications\thanks{\copyright IACR 2016. This article is the final
    version submitted by the authors to the IACR and to
    Springer-Verlag on June 3, 2016. % The version published by
    % Springer-Verlag is available at DOI: 10.1007/978-3-662-53015-3\_2.
  }}

\author{Frédéric Dupuis\inst{2}\and Serge Fehr\inst{1}\and Philippe
  Lamontagne\inst{3}\and Louis Salvail\inst{3}}

\institute{CWI, Amsterdam, The Netherlands \and Faculty of
  Informatics, Masaryk University, Brno, Czech Republic \and
  Université de Montréal (DIRO), Montréal, Canada}

\maketitle
\begin{abstract}
  We prove a general relation between {\em adaptive} and {\em
    non-adap\-tive} strategies in the quantum setting, i.e., between
  strategies where the adversary can or cannot adaptively base its
  action on some auxiliary quantum side information. Our relation
  holds in a very general setting, and is applicable as long as we can
  control the bit-size of the side information, or, more generally,
  its ``information content''. Since adaptivity is notoriously
  difficult to handle in the analysis of (quantum) cryptographic
  protocols, this gives us a very powerful tool: as long as we have
  enough control over the side information, it is sufficient to
  restrict ourselves to non-adaptive attacks.

  We demonstrate the usefulness of this methodology with two
  examples. The first is a quantum bit commitment scheme based on {\em
    1-bit cut-and-choose}. Since bit commitment implies oblivious
  transfer (in the quantum setting), and oblivious transfer is
  universal for two-party computation, this implies the universality
  of $1$-bit cut-and-choose, and thus solves the main open problem
  of~\cite{fkszz13}. The second example is a quantum bit commitment
  scheme proposed in 1993 by Brassard {\em et al}. It was originally
  suggested as an unconditionally secure scheme, back when this was
  thought to be possible. We partly restore the scheme by proving it
  secure in (a variant of) the bounded quantum storage model.

  In both examples, the fact that the adversary holds quantum side
  information obstructs a direct analysis of the scheme, and we
  circumvent it by analyzing a non-adaptive version, which can be done
  by means of known techniques, and applying our main result.
\end{abstract}
%\keywords{(non-)adaptive attacks, quantum cryptography, cut-and-choose, noisy-storage-model }

\section{Introduction}
\label{sec:intro}

\subsubsection{Adaptive Versus Non-Adaptive Attacks. }

We consider attacks on cryptographic schemes, and we compare adaptive
versus non-adaptive strategies for the adversary. In our context, a
strategy is {\em adaptive} if the adversary's action can depend on
some auxiliary side information, and it is {\em non-adaptive} if the
adversary has no access to any such side information. Non-adaptive
strategies are typically much easier to analyze than adaptive ones.

Adaptive strategies are clearly more powerful than non-adaptive ones,
but this advantage is limited by the amount and quality of the
side-information available to the attacker. In the classical case,
this can be made precise by the following simple argument.  If the
side information consists of a classical $n$-bit string, then
adaptivity increases the adversary's success probability in breaking
the scheme by at most a factor of $2^n$. Indeed, a particular
non-adaptive strategy is to try to guess the $n$-bit side information
and then apply the best adaptive strategy. Since the guess will be
correct with probability at least $2^{-n}$, it follows that $\nonadapt
\geq 2^{-n} \adapt$, and thus $\adapt \leq 2^n \nonadapt$, where
$\adapt$ and $\nonadapt$ respectively denote the optimal adaptive and
non-adaptive success probabilities for the adversary to break the
scheme. Even though there is an exponential loss, this is a very
powerful relation between adaptive and non-adaptive strategies as it
applies very generally, and it provides a non-trivial bound as long as
we can control the size of the side information, and the non-adaptive
success probability is small enough.

\subsubsection{Our Technical Result. }

In this work, we consider the case where the side information (and the
cryptographic scheme as a whole) may be {\em quantum}. A natural
question is whether the same (or a similar) relation holds between
adaptive and non-adaptive quantum strategies. The quantum equivalent
to guessing the side information would be to emulate the $n$-qubit
quantum side information by the completely mixed state
$\frac{\id_A}{2^n}$.  Since it always holds that $\rho_{AB}\leq
2^{2n}\frac{\id_A}{2^n}\otimes \rho_B$, we immediately obtain a
similar relation $\adapt \leq 2^{2n} \nonadapt$, but with an
additional factor of $2$ in the exponent. The bound is tight for
certain choices of $\rho_{AB}$, and thus this additional loss is
unavoidable in general; this seems to mostly answer the above
question.

In this work, we show that this is actually not yet the end of the
story. Our main technical result consists of a more refined
treatment\,---\,and analysis\,---\, of the relation between adaptive
and non-adaptive quantum strategies.  We show that in a well-defined
and rather general context, we can actually bound $\adapt$ as $$\adapt
\leq 2^{\imaxacc(B;A)} \nonadapt \enspace ,$$ where $\imaxacc(B;A)$ is
a new (quantum) information measure that is upper bounded by the
number of qubits of $A$. As such, we not only recover the classical
relation $\adapt \leq 2^n \nonadapt$ in the considered context, but we
actually improve on it.

In more detail, we consider an abstract ``game'', specified by an
arbitrary bipartite quantum state $\rho_{AB}$, of which the adversary
Alice and a challenger Bob hold the respective registers $A$ and $B$,
and by an arbitrary family $\{E^j\}_{j \in \cal J}$ of binary-outcome
POVMs acting on register $B$. The game is played as follows: Alice
chooses an index $j$, communicates it to Bob, and Bob measures his
state $B$ using the POVM $E^j = \{E_0^j,E_1^j\}$ specified by
Alice. Alice wins the game if Bob's measurement outcome is~$1$. In the
adaptive version of the game, Alice can choose the index $j$ by
performing a measurement on $A$; in the non-adaptive version, she has
to decide upon $j$ without resorting to $A$.  As we will see, this
game covers a large class of quantum cryptographic schemes, where
Bob's binary measurement outcome specifies whether Alice succeeded in
breaking the scheme.

Our main result shows that in any such game it holds that $\adapt \leq
2^n \nonadapt$ where $n = \hmax{A}$, i.e., the number of qubits of
$A$. Actually, as already mentioned, we show a more general and
stronger bound $\adapt \leq 2^{\imaxacc(B;A)} \nonadapt$ that also
applies if we have no bound on the number of qubits of $A$, but we
have some control over its ``information content'' $\imaxacc(B;A)$,
which is a new information measure that we introduce and show to be
upper bounded by $\hmax{A}$.

To give a first indication of the usefulness of our result, we observe
that it easily provides a lower-bound on the quantity, \emph{or
  quality}, of entanglement (as measured by $\imaxacc(B;A)$) that a
dishonest committer needs in order to carry out the standard
attack~\cite{PhysRevLett.78.3414} on a quantum bit commitment
scheme. Let Alice be the committer and Bob the receiver in a bit
commitment scheme in which the opening phase consists of Alice
announcing a classical string $j$ and Bob applying a verification
described by POVM $\{E_{\mathrm{accept}}^j,
E_{\mathrm{reject}}^j\}$. In the standard attack, Alice always commits
to 0 while purifying her actions and applies an operation on her
register if she wants to change her commitment to~$1$.  If we let
$\rho_{AB}$ be the state of Bob's register $B$ that corresponds to a
commitment to $0$, then the probability that a memoryless Alice
successfully changes her commitment to $1$ is $\nonadapt=
\max_j\trace{E_{\mathrm{accept}}^j \rho_{AB}}$ where the maximum is
over all $j$ that open~$1$. If Alice holds a register $A$ entangled
with $B$, our main result implies that $\imaxacc(B;A)$ must be
proportional to $-\log\nonadapt$ for Alice to have a constant
probability of changing her commitment.

But the real potential lies in the observation that adaptivity is
notoriously difficult to handle in the analysis of cryptographic
protocols, and as such our result provides a very powerful tool: as
long as we have enough control over the side information, it is
sufficient to restrict ourselves to non-adaptive attacks.

\subsubsection{Applications. }

We demonstrate the usefulness of this methodology by proving the
security of two commitment schemes. In both examples, the fact that
the adversary holds quantum side information obstructs a direct
analysis of the scheme, and we circumvent it by analyzing a
non-adaptive version and applying our general result.

\paragraph{One-bit cut-and-choose is universal for two-party computation. }

As a first example, we propose and prove secure a quantum bit
commitment scheme that uses an ideal {\em 1-bit cut-and-choose}
primitive $\onecc$ (see~Fig.~\ref{fig:cc} in Sect.~\ref{sec:onecc}) as
a black box. Since bit commitment ($\bc$) implies oblivious transfer
($\ot$) in the quantum setting~\cite{bbcs91,c94,u09}, and oblivious
transfer is universal for two-party computation, this implies the
universality of $\onecc$ and thus completes the zero/xor/one law
proposed in~\cite{fkszz13}. Indeed, it was shown in~\cite{fkszz13}
that in the information-theoretic quantum setting, every primitive is
either trivial (zero), universal (one), or can be used to implement an
XOR\,---\,{\em except} that there was one missing piece in their
characterization: it excluded $\onecc$ (and any primitive that implies
$\onecc$ but not~$\twocc$). How $\onecc$ fits into the landscape was
left as an open problem in~\cite{fkszz13}; we resolve it here.
    
\paragraph{The BCJL bit commitment scheme in (a variant of) the
  bounded quantum storage model. }
As a second application, we consider a general class of
non-interactive commitment schemes and we show that for any such
scheme, security against an adversary with no quantum memory {\em at
  all} implies security in a slightly strengthened version of the
standard bounded quantum storage model\footnote{Beyond bounding the
  adversary's quantum memory, we also restrict its measurements to be
  projective; this can be justified by the fact that to actually
  implement a non-projective measurement, additional quantum memory is
  needed. }, with a corresponding loss in the error parameter.%
\footnote{We have already shown above how to argue for the standard
  attack~\cite{PhysRevLett.78.3414} against quantum bit commitment
  schemes; taking care of {\em arbitrary} attacks is more involved. }

As a concrete example scheme, we consider the classic BCJL scheme that
was proposed in 1993 by Brassard {\em et al.}~\cite{366851} as a
candidate for an unconditionally-secure scheme\,---\,back when this
was thought to be possible\,---\,but until now has resisted any
rigorous \emph{positive} security analysis.  Our methodology of
relating adaptive to non-adaptive security allows us to prove it
secure in (a variant of) the bounded quantum storage model.

\section{Preliminaries}
\label{sec:prelims}

\subsection{Basic Notation}
%\label{sec:notation}

For any string $x = (x_1,\ldots,x_n) \in \bool^n$ and any subset
$t=\{t_1,\dots t_k\}\subseteq [n]$, we write $x_t$ for the substring
$x_t = (x_{t_1},\dots,x_{t_k}) \in \bool^{|t|}$. The $n$-bit all-zero
string is denoted as $0^n$. The Hamming distance between between two
strings $x,y\in \bool^n$ is defined as $d(x,y)= \sum_{i=1}^n x_i\oplus
y_i$.  For $\delta>0$ and $x\in \bool^n$, $B^\delta(x)$ denotes the
set of all $n$ bit strings at Hamming distance at most $\delta n$
from~$x$.  We denote by $\lg{(\cdot)}$ the logarithm with respect to
base $2$.  It is well known that the set $B^\delta(x)$ contains at
most $2^{nh(\delta)}$ strings where $h(\delta)=
-\delta\lg(\delta)-(1-\delta)\lg(1-\delta)$ is the binary entropy
function.

Ideal cryptographic \emph{functionalities} (or \emph{primitives}) are
referenced by their name written in sans-serif font. They are fully
described by their input/output behaviour (see, e.g., functionality
$\onecc$ described in Fig.~\ref{fig:cc} in
Sect.~\ref{sec:onecc}). Cryptographic \emph{protocols} have their
names written in small capitals with a primitive name in superscript
if the protocol has black-box access to this primitive (e.g. protocol
$\textsc{bc}^\onecc$ in Sect.~\ref{sec:onecc}).

\subsection{Quantum States and More}

We assume familiarity with the basic concepts of quantum information;
we merely fix notation and terminology here.  We label quantum
registers by capital letters $A, B$ etc. and their corresponding
Hilbert spaces are respectively denoted by $\hilbert_A, \hilbert_B$
etc.  We say that a quantum register $A$ is ``empty'' if
$\dim(\hilbert_A) = 1$.  The state of a quantum register is specified
by a density operator $\rho$, a positive semidefinite trace-$1$
operator.  We typically write $\rho_A$ for the state of $A$, etc.  The
set of density operators for register $A$ is denoted ${\cal
  D}(\hilbert_A)$.  We write $X \geq 0$ to express that the operator
$X$ is positive semidefinite, and $Y \geq X$ to express that $Y-X$ is
positive semidefinite.

We measure the distance between two states $\rho$ and $\sigma$ in
terms of their \emph{trace distance} $D(\rho, \sigma):= \frac
12\pnorm{\rho-\sigma}{1}$, where \smash{$\pnorm{X}{1}:=
  \trace{\sqrt{X^\dagger X}}$} is the \emph{trace norm}. We say that
$\rho$ and $\sigma$ are \emph{$\epsilon$-close} if $D(\rho,\sigma)\leq
\epsilon$, and we call them \emph{indistinguishable} if their trace
distance is negligible (in the security parameter).

The \emph{computational} (or rectilinear) basis for a single qubit
quantum register is denoted by $\{\ket 0_+, \ket 1_+\}$, and the
\emph{diagonal} basis by $\{\ket 0_\times, \ket 1_\times\}$.  Recall
that $\ket 0_\times = \frac 1{\sqrt 2} (\ket 0_++\ket 1_+)$ and $\ket
1_\times = \frac 1{\sqrt 2} (\ket 0_+-\ket 1_+)$.  For any $x\in
\bool^n$ and $\theta \in \{+,\times \}^n$, we set $\ket x_\theta :=
\bigotimes_{i=1}^n \ket {x_i}_{\theta_i}$.  In the following, we will
view and represent any sequence of diagonal and computational bases by
a bit string $\theta\in\{0,1\}^n$, where $\theta_i=0$ represents the
computational basis and $\theta_i=1$ the diagonal basis.  In other
words, for $b\in\{0,1\}$, $\ket{b}_0:=\ket{b}_+$ and $\ket{b}_1 :=
\ket{b}_{\times}$. And for $\theta,x\in\{0,1\}^n$, we define
$\ket{x}_\theta := \bigotimes_{i=1}^n \ket{x_i}_{\theta_i}$.

Operations on quantum registers are modeled as completely-positive
trace-preserving (CPTP) maps. To indicate that a CPTP map
$\mathcal{E}$ takes inputs in $A$ and outputs to $B$, we use subscript
$A \rightarrow B$. If $\mathcal E_{A\rightarrow B}$ is a CPTP map
acting on register $A$, we slightly abuse notation and write $\mathcal
E(\rho_{AC})$ instead of $\mathcal E\otimes \id_C(\rho_{AC})$ where
$\id_C$ is the CPTP map that leaves register $C$ unchanged. A
\emph{measurement} on a quantum register $A$, producing a measurement
outcome $X$, is a CPTP map $\mathcal{E}_{A \rightarrow X}$ of the form
\[ \mathcal{E}(\rho_A) = \sum_{x\in \cal X} \trace{E_x \rho_A}
  \proj{x}_X\enspace, \]
where $\{\ket{x}\}$ a basis of $\hilbert_X$ and
$E = \{ E_x \}_{x \in \cal X}$ is a POVM, i.e., a collection of
positive semidefinite operators satisfying
$\sum_{x\in \cal X} E_x= \id$.

The {\em spectral norm} of an operator $X$ is defined as $\|X\|:=
\max_{\ket u}{\|X\ket u\|}$, where the maximum is over all normalized
vectors $\ket u$, and an operator is called an {\em orthogonal
  projector} if $X^\dagger = X$ and $X^2 = X$.  The following was
shown in~\cite{dfss05}.

\begin{lemma}\label{lem:normineq}
  For any two orthogonal projectors $X$ and $Y$: $\|X+Y\|\leq
  1+\|XY\|$. 
\end{lemma}

\subsection{Entropy and Privacy Amplification}\label{sec:uncertrel}
In the following, the two notions of entropy that we will be dealing
with are the min-entropy and the zero-entropy of a quantum
register. They are defined as follows:
\begin{definition}\label{def:min-maxentropy}
  The \emph{min-entropy} of a bipartite quantum state $\rho_{AB}$
  relative to register $B$ is the largest number $\hmin{A|B}_{\rho}$
  such that there exists a $\sigma_B\in {\cal D}(\hilbert_B)$,
  $$2^{-\hmin{A|B}_{\rho}}\cdot \id_A\otimes \sigma_B\geq \rho_{AB}  \enspace.$$
 The \emph{zero-entropy} of a state $\rho_A$ is defined as
  $${\hmax{A}_\rho} = \lg{\left( \rank{\rho_A}\right)}
  \enspace.$$ We write $\hmin{A|B}$ and $\hmax{A}$ when the state of
  the registers is clear from the context.
\end{definition}
The min-entropy has the following operational
interpretation~\cite{min-max-entropy}. Let $\rho_{XB}$ be a so-called
cq-state, i.e., of the from $\rho_{XB} = \sum_x P_X(x)\proj x_X
\otimes \rho^x_B$. Then $\guess(X|B)= 2^{-\hmin{X|B}_\rho}$ where
$\guess(X|B)$ is the probability of guessing the value of the
classical random variable $X$, maximized over all POVMs on $B$.

Let $\mathcal G_n$ be a family of hash functions $g:
\{0,1\}^n\to\{0,1\}$ with a binary output.  The family $\mathcal G_n$
is said to be \emph{two-universal} if for any $x, y\in \bool^n$ with
$x\neq y$ and $G\in_R\mathcal G_n$, $$\Pr{\left(G(x)=G(y)\right)}\leq
\frac 12\enspace.$$

Privacy amplification against quantum side information, in case of
hash functions with a binary-output, can be stated as follows:

\begin{theorem}[Privacy amplification~\cite{RK04}]\label{thm:privacyamp}
  Let $\mathcal G_n$ be a two-universal family of hash functions $g:
  \{0,1\}^n\to\{0,1\}$ with a binary output.  Furthermore, let
  $\rho_{XE} = \sum_{x \in \bool^n} P_X(x) \proj{x}_X \otimes
  \rho_E^x$ be an arbitrary cq-state, and let
  \[ \rho_{YGXE} := \frac{1}{|\mathcal{G}_n|} \sum_{g \in
      \mathcal{G}_n} \sum_{x \in \bool^n} P_X(x) \proj{g(x)}_Y \otimes
    \proj{g}_G \otimes \proj{x}_X \otimes \rho_E^x \] be the state
  obtained by choosing a random $g$ in $\mathcal{G}_n$, applying $g$
  to the value stored in $X$, and storing the result in register $Y$.
  Then,
  \[ D\biggl(\rho_{YGE}, \frac{\id_{Y}}2\otimes \rho_{GE}\biggr) \leq
    \frac 12 \cdot 2^{-\frac 12 (\hmin{X|E}-1)}\enspace. \]
\end{theorem}

\section{Main Result}
\label{sec:mainresult}

We consider an abstract game between two parties Alice and Bob. The
game is specified by a joint state $\rho_{AB}$, shared between Alice
and Bob who hold respective registers $A$ and $B$, and by a non-empty
finite family ${\bf E} = \{E^j\}_{j\in \mathcal J}$ of binary-outcome
POVMs $E^j = \{E^j_0,E^j_1\}$ acting on $B$. An execution of the game
works as follows: Alice announces an index $j \in \mathcal J$ to Bob,
and Bob measures register $B$ of the state $\rho_{AB}$ using the POVM
$E^j$ specified by Alice's choice of $j$. Alice {\em wins} the game if
the measurement outcome is $1$. We distinguish between an {\em
  adaptive} and a {\em non-adaptive} Alice. An {\em adaptive} Alice
can obtain $j$ by performing a measurement on her register $A$ of
$\rho_{AB}$; on the other hand, an {\em non-adaptive} Alice has to
produce $j$ from scratch, i.e., without accessing $A$.  This motivates
the following formal definitions.

\begin{definition}
  Let $\rho_{AB}$ be a bipartite quantum state, and let ${\bf E} =
  \{E^j\}_{j\in \mathcal J}$ be a non-empty finite family of
  binary-outcome POVMs $E^j = \{E^j_0,E^j_1\}$ acting on $B$.  Then,
  we define
$$
\succ(\rho_{AB},{\bf E}) := \max_{\{F_j\}_{j}}\sum_{j\in \mathcal J}
\tr\Bigl(\bigl(F_j\otimes E_1^j\bigr)\rho_{AB}\Bigr) \enspace ,
$$
where the maximum is over all POVMs $\{F_j\}_{j\in \mathcal J}$ acting
on $A$.  We call \linebreak$\succ(\rho_{AB},{\bf E})$ the
\emph{adaptive success probability}, and we call $\succ(\rho_{B},{\bf
  E})$ the \emph{non-adaptive success probability}, where the latter
is naturally understood by considering an ``empty'' $A$, and it equals
$$
\succ(\rho_{B},{\bf E}) = \max_{j\in \cal J}\tr\bigl(E_1^j
\rho_{B}\bigr) \enspace .
$$
If $\rho_{AB}$ and $\bf E$ are clear from the context, we write
$\adapt$ and $\nonadapt$ instead of $\succ(\rho_{AB},{\bf E})$ and
$\succ(\rho_{B},{\bf E})$.
\end{definition}

As a matter of fact, for the sake of generality, we consider a setting
with an additional quantum register $A'$ to which both the adaptive
and the non-adaptive Alice have access to, but, as above only the
adaptive Alice has access to $A$. In that sense, we will compare an
adaptive with a {\em semi-adaptive} Alice.  Formally, we will consider
a tripartite state $\rho_{AA'B}$ and relate $\succ(\rho_{AA'B},{\bf
  E})$ to $\succ(\rho_{A'B},{\bf E})$. Obviously, the special case of
an ``empty'' $A'$ will then provide a relation between $\adapt$ and
$\nonadapt$.

We now introduce a new measure of (quantum) information
$\imaxacc(B;A|A')_\rho$, which will relate the adaptive to the non- or
semi-adaptive success probability in our main theorem. In its
unconditional form $\imaxacc(B;A)_\rho$, it is the accessible version
of the max-information $\imax(B;A)_\rho$ introduced
in~\cite{berta-christ-renn-reverse}; this means that it is the amount
of max-information that can be accessed via measurements on Alice's
share.
 
\begin{definition} \label{def:hmaxacc} Let $\rho_{AA'B}$ be a
  tripartite quantum state. Then, we define
  \linebreak$\imaxacc(B;A|A')_{\rho}$ as the smallest real number such
  that, for every measurement $\mathcal{M}_{AA' \rightarrow X}$ there
  exists a measurement $\mathcal{N}_{A' \rightarrow X}$ such that
    \[ \mathcal{M}(\rho_{AA'B}) \leq 2^{\imaxacc(B;A|A')_{\rho}}
\mathcal{N}(\rho_{A'B}) \enspace. \]
The unconditional version $\imaxacc(B;A)_{\rho}$ is naturally defined
by considering $A'$ to be ``empty''; the above condition then coincides with
$$\mathcal M(\rho_{AB})\leq 2^{\imaxacc(B;A)_\rho}\sigma_X\otimes
\rho_B\enspace ,$$ for some normalized density matrix $\sigma_X \in
{\cal D}(\hilbert_X)$, which can be interpreted as the outcome of a
measurement $\mathcal N_{\mathbb C\rightarrow X}$ on an ``empty''
register.
\end{definition}

We are now ready to state and prove our main result.

\begin{theorem} \label{thm:forcedresult} Let $\rho_{AA'B}$ be a
  tripartite quantum state, and let ${\bf E} = \{E^j\}_{j\in \mathcal
    J}$ be a non-empty finite family of binary-outcome POVMs $E^j$
  acting on $B$. Then, we have that
\[ 
  \succ(\rho_{AA'B},{\bf E}) \leq 2^{\imaxacc(B;A|A')_{\rho}}
  \succ(\rho_{A'B},{\bf E}) \enspace.
 \]
\end{theorem}

By considering an ``empty'' $A'$, we immediately obtain the following.
\begin{corollary}\label{cor:forcedresult} 
  Let $\rho_{AB}$ be a bipartite quantum state, and let ${\bf E} =
  \{E^j\}_{j\in \mathcal J}$ be as above. Then,
\[ 
\adapt \leq 2^{\imaxacc(B;A)_{\rho}} \nonadapt \enspace.
 \]
\end{corollary}

\begin{proof}[of Theorem~\ref{thm:forcedresult}] Let $\{F_j\}_{j \in
    \cal J}$ be an arbitrary POVM acting on $AA'$, and let
  $\mathcal{M}_{AA' \rightarrow J}$ be the corresponding measurement
  $\mathcal{M}(\sigma_{AA'}) = \sum_j \trace{F_j \sigma} \proj{j}$.
  We define the map
$$
\mathcal{E}_{JB \rightarrow \mathbb C}(\sigma_{JB}) := \sum_j
\trace{(\proj{j} \otimes E_1^j) \sigma_{JB}}\enspace,
$$ 
which is completely positive (but not trace-preserving in
general). From the definition of $\imaxacc$, we know that there exists
a measurement $\mathcal{N}_{A' \rightarrow J}$, i.e., a CPTP map of
the form $\mathcal{N}(\sigma_{A'}) = \sum_j \trace{F'_j \sigma}
\proj{j}$ for a POVM $\{F'_j\}_{j \in \cal J}$ acting on $A'$, such
that
$$
     \mathcal{M}(\rho_{AA'B}) \leq 2^{\imaxacc(B;A|A')_{\rho}}
\mathcal{N}(\rho_{A'B}) \enspace. 
$$
Applying $\mathcal{E}$ on both sides gives
$$
(\mathcal{E} \circ \mathcal{M})(\rho_{AA'B}) \leq
2^{\imaxacc(B;A|A')_{\rho}} (\mathcal{E} \circ
\mathcal{N})(\rho_{A'B})\enspace,
$$
and expanding both sides using the definitions of $\mathcal{E}$,
$\mathcal{M}$ and $\cal N$ gives 
\begin{align*}
  \sum_j \trace{(F_j \otimes E_1^j) \rho_{AA'B}} &\leq
  2^{\imaxacc(B;A|A')_{\rho}} \sum_j \trace{(F'_j \otimes E_1^j) \rho_{A'B}} \\&\leq
  2^{\imaxacc(B;A|A')_{\rho}} \succ(\rho_{A'B},{\bf E}) \enspace. 
\end{align*}
This yields the theorem statement, since the left-hand side equals to
\linebreak$\succ(\rho_{AA'B},{\bf E})$ when maximized over the choice
of the POVM $\{F_j\}_{j \in \cal J}$.\qed
\end{proof}
By the following proposition, we see that
Corollary~\ref{cor:forcedresult} implies a direct generalization of
the classical bound, which ensures that giving access to $n$ bits
increases the success probability by at most~$2^n$, to qubits.

% \subsection{Properties of the max-accessible-information}
% \label{sec:propimax}

\begin{proposition}\label{lem:imaxlesshzero}
  For any $\rho_{AB}$, we have that $\imaxacc(B;A)_{\rho} \leq H_0(A)_{\rho}$.
\end{proposition}

\begin{proof}
  Let $\ket{\psi}_{ABR}$ be a purification of $\rho_{AB}$ and let
  $\mathcal{M}_{A\rightarrow X}$ be a measurement on $A$. Since
  $\ket{\psi}$ is also a purification of $\rho_{A}$, there exists a
  linear operator $V_{\bar{A} \rightarrow BR}$ from a register
  $\bar{A}$ of the same dimension as $A$ into $BR$ such that
  $\ket{\psi}_{ABR} = (\id_A \otimes V) \ket{\Phi}_{A\bar{A}}$,
  with $\ket{\Phi} = \sum_i \ket{i}_{A} \otimes
  \ket{i}_{\bar{A}}$. Now, first note that
  \[ 2^{-H_0(A)} (\mathcal{M} \otimes \id)(\Phi_{A\bar{A}}) =
    \sum_x \lambda_x \proj{x}_X \otimes \omega_{\bar{A}}^x \leq
    \sum_x \lambda_x \proj{x}_X \otimes \id_{\bar{A}}\enspace, \]
  where $\{ \lambda_x \}$ is a probability distribution, and each
  $\omega_{\bar{A}}^x$ is normalized because
  $\trace{\Phi}=2^{\hmax{A}}$. Multiplying both sides of the
  inequality by $2^{H_0(A)}$ and conjugating by $V$, we get
  \[ (\mathcal{M} \otimes \id)(\proj{\psi}) \leq 2^{H_0(A)}
    \sum_x \lambda_x \proj{x} \otimes VV^{\dagger}\enspace. \]
  Using the fact that
  $VV^{\dagger} = \psi_{BR}:=\trace[A]{\proj \psi}$, this yields
  \[ (\mathcal{M} \otimes \id)(\proj{\psi}) \leq 2^{H_0(A)}
    \sum_x \lambda_x \proj{x} \otimes \psi_{BR}\enspace. \] Tracing out $R$ on
  both sides and defining $\sigma_X = \sum_x \lambda_x \proj{x}$ then
  yields
\[ (\mathcal{M} \otimes \id)(\rho_{AB}) \leq 2^{H_0(A)}
  \sigma_X \otimes \rho_{B} \enspace , \]
which proves the claim.  \qed
\end{proof}
One might naively expect that also the conditional version
$\imaxacc(B;A|A')_{\rho}$ is upper bounded by $H_0(A)_{\rho}$,
implying a corresponding statement for a {\em semi-adaptive} Alice:
giving access to $n$ {\em additional} qubits increases the success
probability by at most~$2^n$. However, this is not true, as the
following example illustrates. Let register $B$ contain two random
classical bits, and let $A$ and $A'$ be two qubit registers,
containing one of the four Bell states, and which one it is, is
determined by the two classical bits. Alice's goal is to guess the two
bits. Clearly, $A'$ alone is useless, and thus a semi-adaptive Alice
having access to $A'$ has a guessing probability of at most
$\frac14$. On the other hand, adaptive Alice can guess them with
certainty by doing a Bell measurement on $AA'$.

However, Proposition~\ref{lem:imaxlesshzero} does generalize to the
conditional version in case of a {\em classical} $A'$. 

\begin{proposition}\label{prop:classical-side-info}
For any state $\rho_{ZAB}$ with classical $Z$: $$\imaxacc(B;A|Z)_{\rho}
\leq \max_z \imaxacc(B;A)_{\rho^z} \leq H_0(A)_{\rho}\enspace .$$
\end{proposition}

An additional property of $\imaxacc$ is that quantum operations that
are in tensor product form on registers $A$ and $B$ cannot increase
the max-accessible-information.
\begin{proposition}\label{prop:cptpmaps}
  Let $\mathcal E_{AB\rightarrow A'B'}$ be a CPTP map of the form
  $\mathcal E= \mathcal E^A\otimes \mathcal
  E^B$. Then $$\imaxacc(B';A')_{\mathcal E(\rho)} \leq
  \imaxacc(B;A)_\rho\enspace .$$
\end{proposition}
The proofs the two previous results can be found in
Appendix~\ref{sec:na-analysis}.

\section{Application 1: $\onecc$ is Universal}
\label{sec:onecc}

\subsection{Background}

It is a well-known fact that information-theoretically secure
two-party computation is impossible without assumptions. As a result,
one of the natural questions that arises is: what are the minimal
assumptions required to achieve it? One way to attack this question is
to try to identify the simplest cryptographic primitives which, when
made available in a black-box way to the two parties, allow them to
perform arbitrary two-party computations. We then say that such a
primitive is ``universal''. Perhaps the best known such primitive is
one-out-of-two oblivious transfer (\ot), which has been shown to be
universal by Kilian~\cite{k88}. Since then, the power of various
primitives for two-party computation has been studied in much more
detail \cite{k91,k00,mpr10,km11,mpr12,kraschewski-these}. Recently, it
has been shown in~\cite{mpr10} that every non-trivial two-party
primitive (i.e.~any primitive that cannot be done from scratch without
assumptions) can be used as a black-box to implement one of four basic
primitives: oblivious transfer ($\ot$), bit commitment ($\bc$), an XOR
between Alice's and Bob's inputs, or a primitive called
\emph{cut-and-choose} (\textsf{CC}) as depicted in Fig.~\ref{fig:cc}.

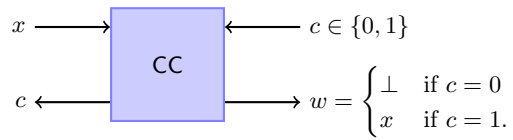
\begin{figure}[h]
    \centering
    \begin{tikzpicture}[thick]
        \node[draw=blue!50, fill=blue!20, minimum height=1.5cm, minimum width=1.5cm] (box) {$\sf{CC}$};

        \coordinate (inputlevel) at ([yshift=.5cm] box.center);
        \coordinate (outputlevel) at ([yshift=-.5cm] box.center);

        \draw 
            (box.west |- inputlevel)    ++(-1,0)    node[left] (x) {$x$}
            (box.east |- inputlevel)    ++(1,0)     node[right] (c) {$c \in \{0,1\}$}
            (box.west |- outputlevel)   ++(-1,0)    node[left] (c2) {$c$}
            (box.east |- outputlevel)   ++(1,0)     node[right] (w) {$w = \begin{cases} \bot & \text{if $c=0$}\\ x & \text{if $c=1$.} \end{cases}$}
            (x) edge[->] (box.west |- inputlevel)
            (c) edge[->] (box.east |- inputlevel)
            (c2) edge[<-] (box.west |- outputlevel)
            (w) edge[<-] (box.east |- outputlevel)
            ;
    \end{tikzpicture}
    \caption{The cut-and-choose functionality. The one-bit and two-bit versions of the functionality refer to the length of $x$. One player chooses $x$, and the other player chooses whether he wants to see $x$ or not. The first player then learns the choice that was made.}
    \label{fig:cc}
\end{figure}

Interestingly, this picture becomes considerably simpler when we
consider quantum protocols. First, $\bc$ can be used to implement
$\ot$~\cite{bbcs91,c94,u09} and is therefore universal. Furthermore,
as was shown in~\cite{fkszz13}, even a 2-bit cut-and-choose (\twocc)
is universal in the quantum setting, giving rise to what they call a
zero/xor/one law: every primitive is either trivial (zero), universal
(one), or can be used to implement an XOR. However, there was one
missing piece in this characterization: it applies to all
functionalities except those that are sufficient to implement 1-bit
cut-and-choose (\onecc), but not \twocc. In this section, we resolve
this issue by showing that \onecc{} is universal. We do this by
presenting a quantum protocol for bit commitment that uses \onecc{} as
a black box, and we prove its security using our adaptive to
non-adaptive reduction.

\subsection{The Protocol}
\label{sec:oneccprotocol}

\begin{figure}[h]
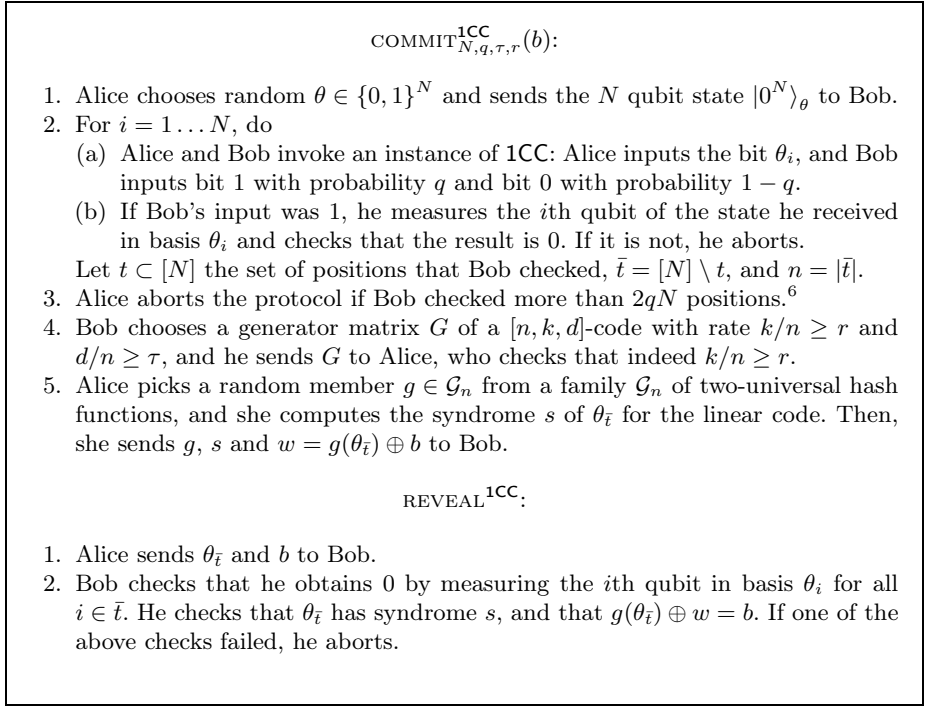

  \begin{framed}
    \begin{center}
      $\textsc{commit}_{N,q,\tau,r}^{\onecc}(b)$:
    \end{center}
    \begin{enumerate}
    \item Alice chooses random $\theta\in \{0,1\}^{N}$ and
      sends the $N$ qubit state $\ket {0^N}_{\theta}$ to Bob. 
    \item For $i=1\dots N$, do \label{st:sampling}
      \begin{enumerate}
      \item Alice and Bob invoke an instance of $\onecc$: Alice inputs
        the bit $\theta_i$, and Bob inputs bit $1$ with probability
        $q$ and bit $0$ with probability $1-q$.
      \item If Bob's input was $1$, he measures the $i$th qubit of the
        state he received in basis $\theta_i$ and checks that the
        result is 0. If it is not, he aborts.
      \end{enumerate}
      Let $t \subset [N]$ the set of positions that Bob checked, $\bar
      t = [N]\setminus t$, and $n = |\bar{t}|$.
    \item Alice aborts the protocol if Bob checked more than $2qN$ 
      positions.\footnotemark
    \item Bob chooses a generator matrix $G$ of a $[n,k,d]$-code with
      rate $k/n \geq r$ and $d/n \geq \tau$, and he sends $G$ to
      Alice, who checks that indeed $k/n \geq r$.
    \item Alice picks a random member $g\in \mathcal G_n$ from a
      family $\mathcal G_n$ of two-universal hash functions, and she
      computes the syndrome $s$ of $\theta_{\bar t}$ for the linear
      code. Then, she sends $g$, $s$ and $w = g(\theta_{\bar t})\oplus
      b$ to Bob.
    \end{enumerate}
    \begin{center}
      $\textsc{reveal}^{\onecc}$:
    \end{center}

  \begin{enumerate}
  \item Alice sends $\theta_{\bar t}$ and $b$ to Bob.
  \item Bob checks that he obtains 0 by measuring the $i$th qubit in
    basis $\theta_i$ for all $i\in \bar t$. He checks that $\theta_{\bar
      t}$ has syndrome $s$, and that $g(\theta_{\bar t})\oplus
    w=b$. If one of the above checks failed, he aborts.
  \end{enumerate}
\end{framed}
\caption{Bit commitment protocol $\textsc{bc}^\onecc$ based on the
  1-bit cut-and-choose primitive.}
\label{fig:commit1cc}
\end{figure}

The protocol is given in Fig.~\ref{fig:commit1cc}, where Alice is the
committer and Bob the receiver. The protocol is parameterized by $N
\in \N$, which acts as security parameter, and by constants $q,\tau$
and $r$, where $q,\tau > 0$ are small and $r < 1$ is close to $1$.
Intuitively, our bit commitment protocol uses the $\onecc$ primitive
to ensure that the state Alice sends to Bob is close to what it is
supposed to be: $\ket{0^N}_\theta$ for some randomly chosen but fixed
basis $\theta$. Indeed, the $\onecc$ primitive allows Bob to sample a
small random subset of the qubits and check for correctness on that
subset; if the state looks correct on this subset, we expect that it
cannot be too far off on the unchecked part.

Note that our protocol uses the B92~\cite{PhysRevLett.68.3121}
encoding ($\{ \ket 0_+, \ket 0_\times\}$), rather than the more common
BB84 encoding. This allows us to get away with a {\em one}-bit
cut-and-choose functionality; with the BB84 encoding, Alice would have
to ``commit'' to {\em two} bits: the basis and the measurement
outcome.

We use the quantum sampling framework of Bouman and
Fehr~\cite{bouman-fehr} to analyze the checking procedure of the
protocol. Actually, we use the {\em adaptive} version
of~\cite{fkszz13}, which deals with an Alice that can decide on the
next basis adaptively depending on what Bob has asked to see so far.
On the other hand, to deal with Bob choosing his sample subset
adaptively depending on what he has seen so far, we require the sample
subset to be rather small, so that we can then apply union bound over
all possible choices.

%%%%% Has to be placed near text that appears on the same page as the
%%%%% figure. This is a cheat to have two footnotes in the same
%%%%% figure.
% \addtocounter{footnote}{-2} %3=n
% \stepcounter{footnote}
\footnotetext{If Bob is honest, this eventuality
  only occurs with probability less than $2\exp(-2q^2N)$ according to
  Hoeffding's inequality.}  

\subsection{Security Proofs}
\label{sec:seca-1cc}

We use the standard notion of hiding for a (quantum) bit commitment scheme. 

\begin{definition}[Hiding]\label{def:security}
  A bit-commitment scheme is $\epsilon$-\emph{hiding} if, for any
  dishonest receiver Bob, his state $\rho_0$ corresponding to a
  commitment to $b = 0$ and his state $\rho_{1}$ corresponding to a
  commitment to $b = 1$ satisfy $D(\rho_0,\rho_1) \leq
  \epsilon$.
\end{definition}

Since the proof that our protocol is hiding uses a standard approach,
we only briefly sketch it.
\begin{theorem}\label{thm:secbob1cc}
  Protocol {\normalfont $\textsc{commit}_{N,q,\tau,r}^{\onecc}$} is
  $2^{-\frac{1}{2} N(\lg(1/\gamma)-2q-(1-r))}$-hiding, where $\gamma =
  \cos^2(\pi/8) \approx 0.85$ (and hence $\lg(1/\gamma) \approx
  0.23$).
\end{theorem}

\begin{proof}[sketch]
  We need to argue that there is sufficient min-entropy in
  $\theta_{\bar t}$ for Bob; then, privacy amplification does the
  job. This means that we have to show that Bob has small success
  probability in guessing~$\theta_{\bar t}$. What makes the argument
  slightly non-trivial is that Bob can choose $t$ depending on the
  qubits $\ket {0^N}_{\theta}$. Note that since Alice aborts in case
  $|t| > 2qN$, we may assume that $|t| \leq 2qN$.

  It is a straightforward calculation to show that Bob's success
  probability in guessing~$\theta$ right after step~1 of the protocol,
  i.e., when given the qubits $\ket {0^N}_{\theta}$, is $\gamma^N$,
  where $\gamma = \cos^2(\pi/8) \approx 0.85$. From this it then
  follows that right after step~\ref{st:sampling}, Bob's success
  probability in guessing~$\theta_{\bar t}$ is at most $\gamma^N \cdot
  2^{2qN}$: if it was larger, then he could guess $\theta$ right after
  step~1 with probability larger than $\gamma^N$ by simulating the
  sampling and guessing the $|t| \leq 2qN$ bits $\theta_i$ that Alice
  provides. It follows that right after step~\ref{st:sampling}, Bob's
  min-entropy in $\theta_{\bar t}$ is $N(\lg(1/\gamma) - 2q)$.
  Finally, by the chain rule for min-entropy, Bob's min-entropy in
  $\theta_{\bar t}$ when additionally given the syndrome $s$ is
  $N\bigl(\lg(1/\gamma) - 2q\bigr) - (n-k) = N\bigl(\lg(1/\gamma) -
  2q\bigr) - n(1-k/n) \geq N\bigl(\lg(1/\gamma) - 2q -
  (1-r)\bigr)$. The statement then directly follows from privacy
  amplification (Theorem~\ref{thm:privacyamp}) and the triangle
  inequality.\qed
\end{proof}

As for the binding property of our commitment scheme, as we will show,
we achieve a strong notion of security that not only guarantees the
existence of a bit to which Alice is bound in that she cannot reveal
the other bit, but this bit is actually {\em universally extractable}
from the classical information held by Bob together with the inputs to
the \onecc:

\begin{definition}[Universally Extractable]\label{def:binding1cc}
  A bit-commitment scheme (in the {\normalfont $\onecc$}-hybrid model)
  is $\epsilon$-\emph{universally extractable} if there exists a
  function $c$ that acts on the classical information {\normalfont
    $view_{Bob,\onecc}$} held by Bob and {\normalfont $\onecc$} after
  the commit phase, so that for any pure commit and open strategy for
  dishonest Alice, she has probability at most $\epsilon$ of
  successfully unveiling the bit {\normalfont
    $1-c(view_{Bob,\onecc})$}.
\end{definition}

Our strategy for proving the binding property for our protocol is as
follows. First, we show that due to the checking part, the (joint)
state after the commit phase is of a restricted form. Then, we show
that, based on this restriction on the (joint) state, a {\em
  non-adaptive} Alice who has no access to her quantum state, cannot
open to the ``wrong'' bit. And finally, we apply our main result to
conclude security against a general (adaptive) Alice.

The following lemma follows immediately from (the adaptive version of)
Bouman and Fehr's quantum sampling
framework~\cite{bouman-fehr,fkszz13}.  Informally, it states that if
Bob did not abort during sampling, then the post-sampling state of
Bob's register is close to the correct state, up to a few errors.  In
other words, after the commit phase, Bob's state is a superposition of
strings close to $0^n$ in the basis specified by $\theta_{\bar t}$.

\begin{lemma}%[Quantum sampling~\cite{bouman-fehr,fkszz13}]
\label{lem:sampling}
Consider an arbitrary pure strategy for Alice in protocol\linebreak
{\normalfont $\textsc{commit}_{N,q,\tau,r}^{\onecc}$} . Let
$\rho_{AB}$ be the joint quantum state at the end of the commit phase,
conditioned (and thus dependent) on $t, \theta, g, w$ and $s$. Then,
for any $\delta>0$, on average over the choices of $t,\theta, g, w$
and $s$, the state $\rho_{AB}$ is $\epsilon$-close to an ``ideal
state'' $\tilde\rho_{AB}$ (which is also dependent on $t, \theta$
etc.) with the property that the conditional state of
$\tilde\rho_{AB}$ conditioned on Bob not aborting is pure and of the
form
 \begin{equation}
   \ket{ \phi_{AB}} =
   \sum_{y\in \Ball(0^n)} \alpha_y\ket{\xi^y}_A\ket y_{\theta_{\bar
       t}} \label{smallsup}
 \end{equation}
 where $\ket{\xi^y}$ are arbitrary states on Alice's register and
 $\epsilon \leq \sqrt{4\exp(-q^2\delta^2N/8)}$.
\end{lemma}

The following lemma implies that after the commit phase, if Alice and
Bob share a state of the form of~(\ref{smallsup}), then a non-adaptive
Alice is bound to a fixed bit which is defined by some string
$\theta'$. 
\begin{lemma}\label{prop:NA1cc}
  For any $t, \theta$ and $s$ there exists
  $\theta'$ with syndrome $s$ such that for every
  $\theta''\neq \theta'$ with syndrome $s$, and for every state $\ket{
    \phi_{AB}}$ of the form of~(\ref{smallsup}),
$$
\tr\bigl((\I \otimes \proj 0_{\theta''}) \phi_{AB}\bigr) \leq
2^{-\frac d2+nh(\delta)} \enspace .
$$
\end{lemma}
\begin{proof}
  Let $\theta'\in\bool^n$ be the string with syndrome $s$ closest to
  $\theta_{\bar t}$ (in Hamming distance). Then, since the set of
  strings with a fixed syndrome form an error correcting code of
  distance $d$, every other $\theta''\in \bool^n$ of syndrome $s$ is
  at distance at least $d/2$ from $\theta_{\bar t}$. Bob's reduced
  density operator of state~(\ref{smallsup}) is $\phi_{B}=
  \sum_{y,y'\in \Ball(0^n)}
  \alpha_y\alpha_{y'}^*\braket{\xi_{y'}}{\xi_{y}} \ketbra
  {y}{y'}_{\theta_{\bar t}}$.  Using the fact that $d(\theta_{\bar
    t},\theta'')\geq d/2$ for every $\theta''\neq \theta'$ (and hence
  $|\trace{\proj{0}_{\theta''} \ketbra{y}{y'}_{\theta_{\bar{t}}}}|
  \leq 2^{-\frac{d}{2}}$) and the triangle inequality, we get:
  \begin{align*}
    \trace{\proj 0_{\theta''} \phi_{B}} 
    &\leq 2^{-\frac{d}{2}} \sum_{y,y' \in \Ball(0^n)} \left|
      \alpha_y \alpha_{y'}^{*} \braket{\xi_{y'}}{\xi_{y}} \right|\\ 
    &\leq 2^{-\frac{d}{2}} \sum_{y,y' \in \Ball(0^n)} |\alpha_y| |\alpha_{y'}^{*}|\\
    &= 2^{-\frac{d}{2}} \bigg( \sum_y |\alpha_y| \bigg)^2\\
    &\leq 2^{-\frac{d}{2} + nh(\delta)}\enspace,
  \end{align*}
  where the last inequality is argued by viewing
  $\sum_{y} |\alpha_y|$ as inner product of the vectors
  $\sum_y |\alpha_y|\ket y$ and $\sum_y\ket y$, and applying the 
  Cauchy-Schwarz inequality. \qed
\end{proof}

We are now ready to prove that the scheme is universally extractable:

\begin{theorem}
  For any $\delta > 0$, {\normalfont
    $\textsc{commit}_{N,q,\tau,r}^{\onecc}$} is $\epsilon$-universally
  extractable with
  $$
  \epsilon \leq 2^{-N(1-2q)(\tau/2-2h(\delta))} + \sqrt{4\exp(-q^2\delta^2N/8)} \enspace.
  $$ 
\end{theorem} 

\begin{proof}
  We need to show the existence of a binary-valued function $c(\theta,
  t, g, w,s)$ as required by Definition~\ref{def:binding1cc}, i.e.,
  such that for any commit strategy, there is no opening strategy that
  allows Alice to unveil $\bar c$, except with small probability.
  We define this function as
  $c(t,\theta,g,s,w):= g(\theta')\oplus w$ where $\theta'$ is as in
  Lemma~\ref{prop:NA1cc}, depending on $t,\theta$ and $s$ only. 

  Now, consider an arbitrary pure strategy for Alice in protocol
  $\textsc{commit}^\onecc$.  Let $\theta, g, w$ and $s$ be the values
  chosen by Alice during the commit phase and let $\rho_{AB}$ be the
  joint state of Alice and Bob after the commit phase.  Fix $\delta >
  0$ and consider the states $\tilde\rho_{AB}$ and $\ket{\phi_{AB}}$
  as promised by Lemma~\ref{lem:sampling}.  Recall that $\rho_{AB}$ is
  $\epsilon$-close to $\tilde\rho_{AB}$ (on average over $\theta, g,
  w$ and $s$, and for $\epsilon \leq \sqrt{4\exp(-q^2\delta^2N/8)}$),
  and $\tilde\rho_{AB}$ is a mixture of Bob aborting in the commit
  phase and of $\ket{\phi_{AB}}$; therefore, we may assume that Alice
  and Bob share the pure state $\phi_{AB} = \proj{\phi_{AB}}$ instead
  of $\rho_{AB}$ by taking into account the probability at most
  $\epsilon$ that the two states behave differently.

  Let $\mathcal B$ be the set of strings $\theta''$ with syndrome $s$
  such that $g(\theta'')\oplus w=\bar c$ and let ${\bf
    E}=\{\{E_0^{\theta''}, E_1^{\theta''}\}\}_{\theta''\in \mathcal
    B}$ be the family of POVMs that correspond to Bob's verification
  measurement when Alice announces $\theta''$, i.e. where
  $E_1^{\theta''}= \proj 0_{\theta''}$ and $E_0^{\theta''}= \id- \proj
  0_{\theta''}$. Then, Alice's probability of successfully unveiling
  bit $\bar c$ equals $\succ(\phi_{AB}, {\bf E})$ as defined in
  Sect.~\ref{sec:mainresult}. In order to apply
  Corollary~\ref{cor:forcedresult}, we must first control the size of
  the side-information that Alice holds. By looking at the definition
  of $\ket{\phi_{AB}}$ in~(\ref{smallsup}), we notice that it is a
  superposition of at most $|\Ball(0^n)|\leq 2^{nh(\delta)}$
  terms. Therefore, the rank of $\phi_A$ is at most $2^{nh(\delta)}$
  and $\hmax{A}\leq nh(\delta)$. We can now bound Alice's probability
  of opening $\bar c$:
  $$ \succ(\phi_{AB}, {\bf E})\leq 2^{\hmax{A}} \succ(\phi_{B},
  {\bf E})\leq 2^{-\frac d2+2nh(\delta)} \leq
  2^{-n(\tau/2-2h(\delta))}$$ where the first inequality follows from
  Corollary~\ref{cor:forcedresult} and
  Proposition~\ref{lem:imaxlesshzero}, and the second from the bound
  on $\hmax{A}$ and from Lemma~\ref{prop:NA1cc}.
  \qed
\end{proof}

%\paragraph{Remark on Bob's choice of $G$.}

Regarding the choice of parameters $q,\tau$ and $r$, and the choice of
the code, we note that the Gilbert-Varshamov bound guarantees that the
code defined by a random binary $n\times(n-rn)$ generator matrix $G$
has minimal distance $d \geq \tau n$, except with negligible
probability, as long as $r< 1-h(\tau)$.  On the other hand, for the
hiding property, we need that $r > 1 - 0.23 + 2q$. As such, as long as
$h(\tau) < 0.23-2q$, there exists a suitable rate $r$ and a suitable
generator matrix $G$, so that our scheme offers statistical security
against both parties.

\subsection{Universality of $\onecc$}\label{sec:UC}

By using our $\onecc$-based bit commitment scheme $\textsc{bc}^{\onecc}$ in the standard
construction for obtaining $\ot$ from $\bc$ in the quantum
setting~\cite{bbcs91,c94}, we can conclude that $\onecc$ implies $\ot$ in the
quantum setting, and since $\ot$ is universal we thus immediately
obtain the universality of $\onecc$. However, strictly speaking, this
does not solve the open problem of~\cite{fkszz13} yet.  The caveat is
that~\cite{fkszz13} asks about the universality of $\onecc$ in the
{\em UC security model}~\cite{u09}, in other words, whether $\onecc$ is
``universally-composable universal''. So, to truly solve the open
problem of~\cite{fkszz13} we still need to argue {\em UC security} of the
resulting $\ot$ scheme, for instance by arguing that our scheme $\textsc{bc}^{\onecc}$ is UC secure. 

UC-security of $\textsc{bc}^{\onecc}$ against malicious Alice follows
immediately from our binding criterion
(Definition~\ref{def:binding1cc}); after the commit phase, Alice is
bound to a bit that can be extracted in a black-box way from the
classical information held by Bob and the $\onecc$
functionality. Thus, a simulator can extract that bit from malicious
Alice and input it into the ideal commitment functionality, and since
Alice is bound to this bit, this ideal-world attack is
indistinguishable from the real-world attack.

However, it is not clear if $\textsc{bc}^{\onecc}$ is UC-secure
against malicious Bob. The problem is that it is unclear whether it is
{\em universally equivocable}, which is a stronger notion than the
standard hiding property (Definition~\ref{def:security}).

Nevertheless, we {\em can} still obtain a UC-secure $\ot$ scheme in
the $\onecc$-hybrid model, and so solve the open problem
of~\cite{fkszz13}. For that, we slightly modify the standard
$\bc$-based $\ot$ scheme~\cite{bbcs91,c94} with $\bc$ instantiated by
$\textsc{bc}^{\onecc}$ as follows: for every BB84 qubit that the
receiver is meant to measure, he commits to the basis using
$\textsc{bc}^{\onecc}$, but he uses the $\onecc$-functionality
{\em directly} to ``commit'' to the measurement outcome, i.e., he
inputs the measurement outcome into $\onecc$\,---\,and if the sender
asks $\onecc$ to reveal it, the receiver also unveils the accompanying
basis by opening the corresponding commitment.

Definition~\ref{def:binding1cc} ensures universal extractability of
the committed bases and thus of the receiver's input. This implies
UC-security against dishonest receiver.  In order to argue UC-security
against dishonest sender, we consider a simulator that acts like the
honest receiver, i.e., chooses random bases and commits to them, but
only measures those positions that the sender wants to
see\,---\,because the simulator controls the $\onecc$-functionality he
can do that. Then, once he has learned the sender's choices for the
bases, he can measure all (remaining) qubits in the correct basis, and
thus reconstruct {\em both} messages and send them to the ideal $\ot$
functionality.  The full details of the proof are in
Appendix~\ref{sec:ucsec}.

\section{Application 2: \\On the Security of \textsc{bcjl} Commitment
  Scheme }
\label{sec:bqsm}

In this section, we show that for a wide class of bit-commitment
schemes, the binding property of the scheme in (a slightly
strengthened version of) the {\em bounded-quantum-storage model}
reduces to its binding property against a dishonest committer that has
{\em no quantum memory at all}. We then demonstrate the usefulness of
this on the example of the \textsc{bcjl} commitment
scheme~\cite{366851}.

\subsection{Setting Up the Stage}

The class of schemes to which our reduction applies consists of the
schemes that are non-interactive: all communication goes from Alice,
the committer, to Bob, the verifier. Furthermore, we require that
Bob's verification be ``projective'' in the following sense.

\begin{definition}
  We say that a bit-commitment scheme is \emph{non-interactive and
    with projective verification}, if it is of the following form.
  \begin{description}\setlength{\parskip}{0.5ex}
  \item{\em Commit:} Alice sends a classical message
    $x$ and a quantum register $B$ to Bob.
  \item{\em Opening to $b$:} Alice sends a classical opening
    $y_b$ to Bob, and Bob applies a binary-outcome
    projective measurement $\{\mathbb V_{x,y_b}, \id-\mathbb
    V_{x,y_b}\}$ to register $B$.
  \end{description}
\end{definition}
Since $x$ is fixed after the commit phase, we tend to leave the
dependency of $\mathbb V_{x,y_b}$ from $x$ implicit and write $\mathbb
V_{y_b}$ instead. Also, to keep language simple, we will just speak of
a {\em non-interactive} bit-commitment scheme and drop the {\em
  projective verification} part in the terminology.

We consider the security\,---\,more precisely: the binding
property\,---\,of such bit-commitment schemes in a slightly
strengthened version of the bounded-quantum-storage
model~\cite{dfss05}, where we bound the quantum memory of Alice, but
we also restrict her measurement (for producing $y_b$ in the opening
phase) to be {\em projective}.  This restriction on Alice's
measurement is well justified since a general non-projective
measurement requires additional quantum storage in the form of an
ancilla to be performed coherently. From a technical perspective, this
restriction (as well as the restriction on Bob's verification) is a
byproduct of our proof technique, which requires the measurement
operator describing the (joint) opening procedure to be repeatable;
avoiding it is an open question.%
\footnote{The standard technique (using Naimark's dilation theorem)
  does not work here. }

Formally, we capture the binding property as follows in this variation
of the bounded-quantum-storage model.%

\begin{definition}[Binding]\label{def:noninteractive}
  A non-interactive bit commitment scheme is called
  \emph{$\epsilon$-binding against $q$-quantum-memory-bounded} (or
  {\em $q$-QMB} for short) {\em projective adversaries} if, for all
  states $\rho_{AB}\in \mathcal D(\hilbert_A\otimes \hilbert_B)$ with
  $\dim(\hilbert_A)\leq 2^q$ and for all classical messages $x$,
  $$ P^A_0(\rho_{AB}) + P^A_1(\rho_{AB}) \leq 1+ \epsilon $$
  where
  $$P_b^{A}(\rho_{AB}) := \max_{ \{\mathbb F_{y_b}\}_{y_b}}\sum_{y_b}
  \trace{(\mathbb F_{y_b}\otimes \mathbb V_{x,y_b})\rho_{AB}} $$
  is the probability of successfully opening bit $b$, maximized over
  all {\em projective} measurements $\{\mathbb F_{y_b}\}_{y_b}$. 
\end{definition}
In case $q = 0$, where the above requirement reduces to 
$$ 
P^{NA}_0(\rho_{AB}) + P^{NA}_1(\rho_{AB}) \leq 1+ \epsilon 
\quad\text{with}\quad
P_b^{NA}(\rho_{AB}) :=\max_{y_b}\trace{\mathbb V_{x,y_b} \rho_{B}}
$$
and $\rho_B= \trace[A]{\rho_{AB}}$, we also speak of
$\epsilon$-binding against \emph{non-adaptive} adversaries.

\subsubsection{On the binding criterion for non-interactive commitment
  schemes}
\label{sec:bindingdiscussion}

Binding criteria analogous to the one specified in
Definition~\ref{def:noninteractive} have traditionally been weak
notions of security against dishonest committers for quantum
commitment schemes, as opposed to criteria that are more in the spirit
of a bit that cannot be opened by the adversary.  While more
convenient for proving security of commitment schemes, a notable flaw
of the $p_0+p_1\leq 1+\epsilon$ definition is that it does not rule
out the following situation. An adversary might, by some complex
measurement, either completely ruin its capacity to open the
commitment, or be able to open the bit of its choice. Then the total
probability of opening 0 and 1 sum to 1, but, conditioned on the
second outcome of this measurement, they sum to 2.  This is obviously
an undesirable property of a quantum bit-commitment scheme.

Non-interactive schemes that are secure according to
Definition~\ref{def:noninteractive} are binding in a stronger
sense. For instance, the above problem of the $p_0+p_1\leq 1+\epsilon$
definition does not hold for non-interactive schemes. If a scheme is
$\epsilon$-binding, then any state $\rho$ obtained by conditioning on
some measurement outcome must satisfy $P^A_0(\rho)+P_1^A(\rho)\leq
1+\epsilon$. If the total probability of opening 0 and 1 was any
higher, then the adversary could have prepared the state $\rho$ in the
first place, contradicting the fact that the protocol is
$\epsilon$-binding. It remains an open question how to accurately
describe the security of non-interactive commitment schemes that
satisfy Definition~\ref{def:noninteractive}.

\subsection{The General Reduction}

We want to reduce security against a $q$-QMB projective adversary to
the security against a non-adaptive adversary (which should be much
easier to show) by means of applying our general
adaptive-to-non-adaptive reduction. However,
Corollary~\ref{cor:forcedresult} does not apply directly; we need some
additional gadget, which is in the form of the following lemma. It
establishes that if there is a commit strategy for Alice so that the
cumulative probability of opening 0 and 1 exceeds 1 by a
non-negligible amount, then there is also a commit strategy for her so
that she can open 0 {\em with certainty} and 1 with still a
non-negligible
probability.

\begin{lemma}\label{lemma:openboth}
  Let $\rho\in \mathcal D(\hilbert_A\otimes \hilbert_B)$ and
  $\epsilon > 0$ be such that
  $P_0^A(\rho)+P_1^A(\rho) \geq 1+\epsilon$. Then, there exists 
  $\rho^0\in \mathcal D(\hilbert_A\otimes \hilbert_B)$ such that
  $P_0^A(\rho^0)=1$ and $P_1^A(\rho^0)\geq \epsilon^2$.
\end{lemma}
\begin{proof}
  Let $\{\mathbb F_{y_0}\}_{y_0}$ and $\{\mathbb G_{y_1}\}_{y_1}$ be
  the projective measurements maximizing $P_0^A(\rho)$ and
  $P_1^A(\rho)$, respectively. Define the projections onto the
  $0/1$-accepting subspaces as
  $$ \mathbb P_0:= \sum_{y_0} \mathbb F_{y_0}\otimes \mathbb
  V_{y_0} \text{ and } \mathbb P_1:= \sum_{y_1} \mathbb G_{y_1}\otimes \mathbb
  V_{y_1}\enspace .$$

  Since
  $\trace{(\mathbb P_0+\mathbb P_1)\rho}= P_0^A(\rho)+P_1^A(\rho) \geq
  1+\epsilon$,
  it follows that $\|\mathbb P_0+\mathbb P_1\| \geq 1+\epsilon$. From
  Lemma~\ref{lem:normineq}, we have that
  $$ 1+\| \mathbb P_1\mathbb P_0 \| \geq \|\mathbb P_0+\mathbb P_1\|
  \geq 1+\epsilon\enspace .$$
  Therefore there exists $\ket \phi$ such that
  $\| \mathbb P_1\mathbb P_0 \ket \phi\| \geq \epsilon$. Define
  $\ket {\phi_0}:= \mathbb P_0\ket \phi/\|\mathbb P_0\ket {\phi}\|$,
  which we claim has the required properties. The probability to open
  0 from $\ket {\phi_0}$ is $\|\mathbb P_0 \ket {\phi_0}\|^2=1$, and
  the probability to open 1 from $\ket {\phi_0}$ is
  $\|\mathbb P_1\mathbb P_0 \ket {\phi_0}\|^2=\|\mathbb P_1\mathbb P_0
  \ket {\phi}\|^2/\|\mathbb P_0\ket {\phi}\|^2 \geq \epsilon^2$.
\qed
\end{proof}
Now, we are ready to state and prove the general reduction. 
\begin{theorem}\label{thm:non-interactivebinding}
  If a non-interactive bit-commitment scheme is $\epsilon$-binding
  against non-adaptive adversaries, then it is
  $(2^{\frac 12q}\sqrt{\epsilon})$-binding against $q$-QMB projective
  adversaries.
\end{theorem}
\begin{proof}
  Let $\rho_{AB}\in \mathcal D(\hilbert_A\otimes \hilbert_B)$ be the
  joint state of Alice and Bob where $\dim(\hilbert_A)\leq 2^q$ and
  let $\alpha>0$ be such that the opening probabilities satisfy
  $P_0^A(\rho)+P_1^A(\rho)= 1+\alpha$. From
  Lemma~\ref{lemma:openboth}, we know that there exists
  $\rho_{AB}^0\in \mathcal D(\hilbert_A\otimes \hilbert_B)$ constructed
  from $\rho$ such that
  $$P_0^A(\rho^0)=1 \text{ and } P_1^A(\rho^0)\geq\alpha^2\enspace .$$
  We use Corollary~\ref{cor:forcedresult} and the assumption that the
  protocol is $\epsilon$-binding against non-adaptive adversaries to
  show that $\alpha$ cannot be too large.  Let
  $\{\mathbb F_{y_0}\}_{y_0}$ be the measurement that maximizes
  $P_0^A(\rho^0)$.  Let us consider Bob's reduced density operator of
  $\rho^0$:
  \[
    \rho^0_B = \trace[A]{\rho^0_{AB}}
    = \sum_{y_0}\trace[A]{(\mathbb F_{y_0}\otimes \id)\rho^0_{AB}}
    = \sum_{{y_0}}\lambda_{y_0} \sigma_{y_0}
  \]
  where for each ${y_0}$, it holds that
  $\trace{\mathbb V_{y_0} \sigma_{y_0}}=1$. This implies
  $\trace{\mathbb V_{y_1} \sigma_{y_0}}\leq \epsilon$ for every $y_1$
  that opens 1 from our assumption of the non-adaptive security of the
  commitment scheme. Then
  $$ P_1^{NA}(\rho^0_{AB})= \max_{y_1}\trace{\mathbb V_{y_1} \rho^0_{B}}
  = \max_{y_1}\sum_{y_0} \lambda_{y_0} \trace{\mathbb V_{y_1}
    \sigma_{y_0}} \leq \epsilon\enspace .$$
  Applying Corollary~\ref{cor:forcedresult} completes the proof:
  $$\alpha^2 \leq P_1^A(\rho^0)\leq
  2^{\imaxacc(B;A)_{\rho_0}}P_1^{NA}(\rho^0) \leq
  2^{\hmax{A}_{\rho_0}}\epsilon \leq 2^{q}\epsilon
  \enspace .$$\qed
\end{proof}

\subsection{Special Case: the \textsc{bcjl} Bit-Commitment Scheme}

In this subsection, we use the results of the previous section to
prove the security of the \textsc{bcjl} scheme in the bounded
storage model against projective measurement attacks.

The \textsc{bcjl} bit-commitment scheme was proposed in 1993 by
Brassard, Crépeau, Jozsa, and Langlois~\cite{366851}. They proposed to
hide the committed bit using a two-universal family of hash functions
applied on the codeword of an error correcting code and then send this
codeword through BB84 qubits. The idea behind this protocol is that
privacy amplification hides the committed bit while the error
correcting code makes it hard to change the value of this bit without
being detected. While their intuition was correct, their proof
ultimately was not, as shown by Mayers' impossibility result for bit
commitment~\cite{PhysRevLett.78.3414}.

The following scheme differs only slightly from the
original~\cite{366851}, this allows us to recycle some of the analysis
from Sect.~\ref{sec:onecc}.

\begin{figure}[h]
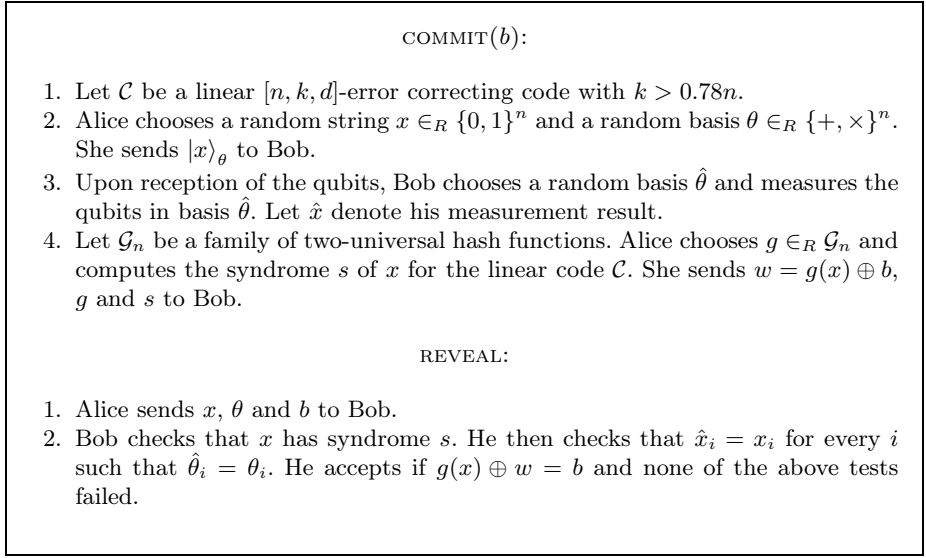

  \begin{framed}
    \begin{center}
      \textsc{commit($b$)}:
    \end{center}

    \begin{enumerate}
    \item Let $\mathcal C$ be a linear $[n,k,d]$-error correcting code
      with $k>0.78n$.
    \item Alice chooses a random string $x\in_R \bool^n$ and a
      random basis $\theta\in_R \{+,\times\}^n$. She sends $\ket
      x_{\theta}$ to Bob.
    \item Upon reception of the qubits, Bob chooses a random basis
      $\hat \theta$ and measures the qubits in basis $\hat
      \theta$. Let $\hat x$ denote his measurement result.
    \item Let $\mathcal G_n$ be a family of two-universal hash
      functions. Alice chooses $g\in_R \mathcal G_n$ and computes the
      syndrome $s$ of $x$ for the linear code $\mathcal C$. She sends
      $w=g(x)\oplus b$, $g$ and $s$ to Bob.
    \end{enumerate}
  %\end{framed}  
  %\begin{framed}  
    \begin{center}
      \textsc{reveal}:
    \end{center}

    \begin{enumerate}
    \item Alice sends $x$, $\theta$ and $b$ to Bob.
    \item Bob checks that $x$ has syndrome $s$. He then checks that
      $\hat x_i=x_i$ for every $i$ such that $\hat \theta_i=
      \theta_i$. He accepts if $g(x)\oplus w= b$ and none of
      the above tests failed.
    \end{enumerate}
  \end{framed}
  \caption{The \textsc{bcjl} bit-commitment scheme}
\label{fig:bc}
\end{figure}

\begin{theorem}\label{thm:noisy-hiding}
  \textsc{bcjl} is statistically hiding as long as
  $0.22 - (1-k/n) \in \Omega(1)$.
\end{theorem}
The proof of Theorem~\ref{thm:noisy-hiding} is straightforward. It
follows the same approach as that of Theorem~\ref{thm:secbob1cc} by
noticing that Bob has the same uncertainty about each $x_i$ as he had
about $\theta_i$ in protocol $\textsc{commit}^\onecc$.

Instead of proving that \textsc{bcjl} is binding, we prove that an
equivalent scheme \textsc{bcjl$_\delta$} (see Fig.~\ref{fig:bcprime})
is binding. The \textsc{bcjl$_\delta$} scheme is a modified version of
\textsc{bcjl} in which Bob has unlimited quantum memory and stores the
qubits sent by Alice during the commit phase instead of measuring
them. The opening phase of \textsc{bcjl$_\delta$} is characterized by
a parameter $\delta$ which determines how close it is to the opening
phase of \textsc{bcjl}. The following lemma shows that the two
protocols are equivalent from Alice's point of view; if Alice can
cheat an honest Bob then she can cheat a Bob with unbounded quantum
computing capabilities.

\begin{figure}[h]
  \begin{framed}
    \begin{center}
      \textsc{commit$_\delta$}($b$):
    \end{center}

    \begin{enumerate}
    \item Let $\mathcal C$ be a linear $[n,k,d]$-error correcting code
      with $k>0.78n$.
    \item Alice chooses a random string $x\in_R \bool^n$ and a
      random basis $\theta\in_R \{+,\times\}^n$. She sends $\ket
      x_{\theta}$ to Bob.
    \item Bob stores all the qubits he received to measure them later.
    \item Let $\mathcal G_n$ be a family of two-universal hash
      functions. Alice chooses $g\in_R \mathcal G_n$ and computes the
      syndrome $s$ of $x$ for the linear code $\mathcal C$. She sends
      $w=g(x)\oplus b$, $g$ and $s$ to Bob.
    \end{enumerate}
    \begin{center}
      \textsc{reveal$_\delta$}:
    \end{center}

    \begin{enumerate}
    \item Alice sends $x$, $\theta$ and $b$ to Bob.
    \item Bob checks that $x$ has syndrome $s$. He then measures his
      stored qubits in basis $\theta$ to obtain $\hat x$. He accepts
      if $d(x,\hat x)\leq \delta n$ and $g(x)\oplus w=b$ and
      none of the above tests failed. \label{step:storedmeasure}
    \end{enumerate}
  \end{framed}
  
  \caption{The \textsc{bcjl$_\delta$} bit-commitment scheme.}
\label{fig:bcprime}
\end{figure}

\begin{lemma} \label{lemma:equiv} Let $\delta>0$. If
  \textsc{bcjl$_\delta$} is $\epsilon$-binding then
  \textsc{bcjl} is $(\epsilon+2\cdot 2^{-\delta n})$-binding.
\end{lemma}
\begin{proof}
  Let $(x,\theta)$ be an opening to 0. First notice that Bob's actions
  in \textsc{bcjl} are equivalent to holding onto his state until the
  opening procedure, measuring in basis $\theta$ and verifying $x_T=
  \hat x_T$ for a randomly chosen sample $T\subseteq [n]$. From this
  point of view, Bob's measurement result is identically distributed
  in both protocols and we can speak of $\hat x$ without ambiguity.
  If $d(x,\hat x)> \delta n$, then the probability that $x_i= \hat
  x_i$ for all $i\in T$ is at most $2^{-\delta n}$.  Therefore, if Bob
  rejects in \textsc{reveal$_\delta$} with measurement outcome $\hat
  x$, then the probability that he rejects in \textsc{reveal} with the
  same outcome is at least $1-2^{-\delta n}$.  If we let $p_0$ denote
  Bob's accepting probability in the original protocol and
  $p_0^\delta$ in the modified protocol, we have $p_0 \leq p_0^\delta
  + 2^{-\delta n}$. Since the same holds for openings to 1, we have
  $$p_0+p_1\leq p_0^\delta+ p_1^\delta + 2\cdot 2^{-\delta n} \leq
  1+\epsilon + 2\cdot 2^{-\delta n} \enspace. $$ \qed
\end{proof}

The following proposition establishes the security of \textsc{bcjl$_\delta$} in
the non-adaptive setting. Its proof is straightforward and can be
found in Appendix~\ref{sec:na-analysis}.
\begin{proposition}\label{lem:probNA}
  \textsc{bcjl$_\delta$} is $2^{-d/2+\delta n +h(\delta)n}$-binding
  against non-adaptive adversaries.
\end{proposition}

Since the bit-commitment scheme \textsc{bcjl$_\delta$} is
non-interactive, it directly follows from
Theorem~\ref{thm:non-interactivebinding} and Proposition~\ref{lem:probNA}
that \textsc{bcjl$_\delta$} is
$2^{\frac 12(q-d/2 +\delta n + h(\delta)n)}$-binding against $q$-QMB
projective adversaries. Combining the above with
Lemma~\ref{lemma:equiv}, we have the following statement for the
\textsc{bcjl} scheme.

\begin{theorem}
  The \textsc{bcjl} bit-commitment scheme is
  $(2^{\frac 12(q-d/2 +\delta n + h(\delta)n)}+2\cdot 2^{-\delta
    n})$-binding against $q$-QMB projective adversaries.
\end{theorem}

\section*{Acknowledgments}
FD acknowledges the support of the Czech Science Foundation (GA ČR)
project no GA16-22211S and of the EU FP7 under grant agreement no
323970 (RAQUEL). LS is supported by Canada's NSERC discovery grant.

\appendix

\section{Additional proofs}
\label{sec:na-analysis}

% custom environment to use custom proposition number 
\newtheorem{innercustomthm}{Proposition}
\newenvironment{customthm}[1]
  {\renewcommand\theinnercustomthm{#1}\innercustomthm}
  {\endinnercustomthm}

\begin{customthm}{\ref{prop:classical-side-info}}
For any state $\rho_{ZAB}$ with classical $Z$: $$\imaxacc(B;A|Z)_{\rho}
\leq \max_z \imaxacc(B;A)_{\rho^z} \leq H_0(A)_{\rho}\enspace .$$
\end{customthm}
\begin{proof}
  By assumption, $\rho_{ZAB}$ is of the form $\rho_{ZAB} = \sum_z
  P_Z(z)\proj{z} \otimes \rho^z_{AB}$. Let ${\cal M}_{ZA\to X}$ be a
  measurement on $Z$ and $A$. By linearity, and by definition of
  $\imaxacc$, we have that
  \begin{align*} {\cal M}(\rho_{ZAB}) &= \sum_z P_Z(z){\cal
      M}\bigl(\proj{z} \otimes \rho^z_{AB}\bigr) \\&\leq \sum_z P_Z(z)
    \cdot 2^{\imaxacc(B;A|Z)_{\proj{z} \otimes \rho^z}}{\cal
      N}^z\bigl(\proj{z} \otimes \rho^z_{B}\bigr)
\end{align*}
for suitably chosen measurements ${\cal N}^z_{Z\to X}$. Now, noting
that $\imaxacc(B;A|Z)_{\proj{z} \otimes \rho^z} =
\imaxacc(B;A)_{\rho^z}$, and that there exists a fixed measurement
${\cal N}_{Z\to X}$ so that \linebreak${\cal N}^z(\proj{z}) = {\cal
  N}(\proj{z})$ for all $z$, it follows that
$$
{\cal M}(\rho_{ZAB}) \leq 2^{\max_z \imaxacc(B;A)_{\rho^z}} {\cal
  N}(\rho_{ZB}) \enspace ,
$$
which implies the first claimed inequality. The second inequality follows immediately by observing that $\imaxacc(B;A)_{\rho^z} \leq H_0(A)_{\rho^z} \leq H_0(A)_{\rho}$. \qed
\end{proof}

\begin{customthm}{\ref{prop:cptpmaps}}% [Restatement of Proposition~\ref{prop:cptpmaps}]
  Let $\mathcal E_{AB\rightarrow A'B'}$ be a CPTP map of the form
  $\mathcal E= \mathcal E^A\otimes \mathcal
  E^B$. Then $$\imaxacc(B';A')_{\mathcal E(\rho)} \leq
  \imaxacc(B;A)_\rho\enspace .$$
\end{customthm}
\begin{proof}
  Since the CPTP map $\mathcal E^B$ commutes with any measurement
  applied on Alice's register, it cannot increase the maximal
  accessible information.

  To show that the CPTP map $\mathcal E^A$ cannot increase $\imaxacc$,
  it suffices to show that for every measurement $\mathcal M$ on
  register $A$, the CPTP map $\mathcal M\circ \mathcal E^A$ is also a
  measurement. Let $\{E_k\}_k$ be the Kraus operators associated with
  $\mathcal E^A$ and let $\{F_x\}_x$ be the POVM operators describing
the measurement $\mathcal M$. Then, the positive operators
$F'_x :=\sum_k E_k^\dagger F_x E_k$ describe a POVM $\mathcal M'$,
and
$$ \mathcal M\circ \mathcal E^A (\rho)= \mathcal M'(\rho)\leq
2^{\imaxacc(B;A)_\rho}\sigma_X\otimes \rho_B$$
by the definition of $\imaxacc(B;A)_\rho$ for some normalized
$\sigma_X$. \qed
\end{proof}

\begin{customthm}{\ref{lem:probNA}}% [Restatement of Proposition~\ref{lem:probNA}]
  Protocol \textsc{bcjl$_\delta$} is
  $2^{-d/2+\delta n +h(\delta)n}$-binding against non-adaptive
  adversaries.
\end{customthm}
\begin{proof}
  Let $\rho_{AB}\in \mathcal D(\hilbert_A\otimes \hilbert_B)$ be the
  joint state of Alice and Bob and let $\mathbb V^\delta_{x,\theta}:=
  \sum_{z\in\Ball(x)} \proj z_\theta$ be the projective measurement
  corresponding to Bob's verification procedure in protocol
  \textsc{bcjl$_\delta$} if Alice announced $(x,\theta)$. Using
  Lemma~\ref{lem:normineq}, we have that for any two distinct openings
  $(x,\theta)$ and $(x',\theta')$,
  \begin{align*}
%    P_0^{NA}(\rho_{AB})+ P_1^{NA}(\rho_{AB})
    \trace{ \mathbb
      V^\delta_{x,\theta}\rho_B} +\trace{\mathbb
    V^\delta_{x',\theta'} \rho_B}
    &=     \trace{ (\mathbb
      V^\delta_{x,\theta}+\mathbb
      V^\delta_{x',\theta'} )\rho_B}
    \\&\leq ||\mathbb V^\delta_{x,\theta}+
        \mathbb V^\delta_{x',\theta'}||
    \\&\leq 1+||\mathbb V^\delta_{x,\theta}
        \mathbb V^\delta_{x',\theta'}||\enspace .
  \end{align*}
  Using techniques from~\cite{bfgs11}, we can show that
  $$||\mathbb V^\delta_{x,\theta} \mathbb V^\delta_{x',\theta'}||\leq
  \max_{\substack{z\in\Ball(x)\\ z'\in\Ball(x')}} |\bra
  z_{\theta}\ket{z'}_{\theta'}|\sqrt{|\Ball(x)||\Ball(x')|}\enspace
  .$$ Using the fact that $d(z,z')\geq d-2\delta n$ for $z\in
  \Ball(x)$ and $z'\in \Ball(x')$ for any two strings $x$ and $x'$
  with the same syndrome, and the fact that $|\Ball(x)|\leq
  2^{h(\delta)n}$, it follows that when maximizing over openings to 0
  and 1, we obtain
  $$P_0^{NA}(\rho_{AB})+ P_1^{NA}(\rho_{AB}) \leq
  1+2^{-d/2+\delta n+h(\delta)n}\enspace .$$ \qed
\end{proof}

\section{UC-Completeness of $\onecc$}
\label{sec:ucsec}

\subsection{The UC Model}

In order to show that a scheme securely implements a given
functionality $\sf F$ in the universally composable (UC) model, one
has to show that for any \emph{adversary} that attacks the scheme by
corrupting participants, there exists a \emph{simulator} $\siml$ that
instead attacks the functionality, but is indistinguishable from the
adversary from an outside observer's perspective. More precisely, one
considers an \emph{environment} $\env$ that interacts with the
adversary in the \emph{real} model where the scheme is executed, or
with $\siml$ in the \emph{ideal} model where the functionality $\sf F$
is executed, and it provides input to and obtains output from the
uncorrupt players (see Fig.~\ref{fig:realidealmodel}). The scheme is
said to {\em statistically quantum-UC-emulate} the functionality if
the environment cannot distinguish the real from the ideal model with
non-negligible probability. For a more detailed description of the
quantum UC framework, we refer to~\cite{fkszz13,u09}.

\begin{figure}[h]
  \begin{center}
    \begin{tikzpicture}[node distance=2cm, semithick]\node[] (alice) {$A$}; \node[right of=alice,
      xshift=2cm] (bob) {$B$}; \node[right of=bob] (env) {$\env$};
      \node[right of=alice, yshift=-0.5cm] (1cc) {$\onecc$};

      \path (alice) edge[bend left=15, <->] (bob) (bob) edge[bend
      left=10, <->] (1cc) (1cc) edge[bend left=10, <->] (alice);

      \path (env.west)+(0,0.1) edge[<-] node[above] {\tiny committed}
      ($(bob.east)+(0,0.1)$) (env.west)+(0,-0.1) edge[<-] node[below]
      {\tiny (open,$b$)} ($(bob.east)+(0,-0.1)$);
      
      \draw[<->, rounded corners] (alice.north) -| +(0,1) -| (env);
      \end{tikzpicture}

\vspace{2ex}

      \begin{tikzpicture}[node distance=2cm, semithick]\node[] (alice)
        {$\siml$}; \node[right of=alice] (bc) {\bc}; \node[right
        of=bc] (bob) {$\env$};

        \path (alice.east)+(0,0.1) edge[->] node[above] {\tiny
          (commit,$b$)} ($(bc.west)+(0,0.1)$) (alice.east)+(0,-0.1)
        edge[->] node[below] {\tiny open} ($(bc.west)+(0,-0.1)$)
        (bob.west)+(0,0.1) edge[<-] node[above] {\tiny committed}
        ($(bc.east)+(0,0.1)$) (bob.west)+(0,-0.1) edge[<-] node[below]
        {\tiny (open,$b$)} ($(bc.east)+(0,-0.1)$);
     
        \draw[<->, rounded corners] (alice.north) -| +(0,1) -| (bob);
    \end{tikzpicture}\vspace{-3ex}
  \end{center}
  \caption{The real model (top) and the ideal model (bottom) for
    protocol $\textsc{bc}^\onecc$ and functionality $\bc$, respectively, with a dishonest Alice. 
    $\textsc{bc}^\onecc$ statistically quantum-UC-emulates $\bc$
    (against dishonest Alice) if the two models are indistinguishable for $\env$.}
  \label{fig:realidealmodel}
\end{figure}
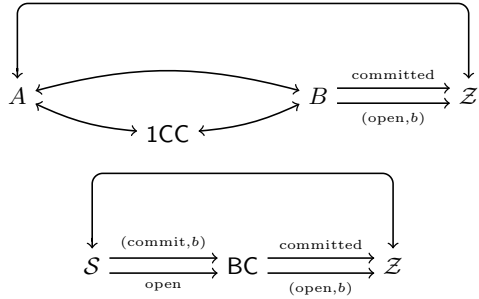

Most UC security proofs follow a similar mold.  $\siml$ internally
runs a copy of the adversary, and it simulates the actions and
interactions of the honest party, and of functionalities that are
possibly used as subroutines in the scheme. $\siml$ must look like the
real model adversary to the environment $\env$, so it forwards any
message it receives from $\env$ to (its internal execution of) the
adversary and vice versa.  Furthermore, from the interaction with the
adversary, it extracts the input(s) it has to provide to $\sf F$ (see
Fig.~\ref{fig:typicalsim}).

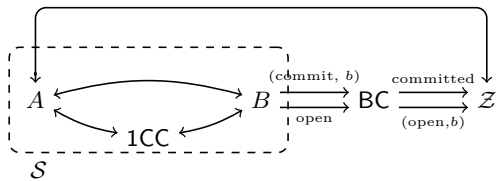
\begin{figure}[h]
  \begin{center}
    \begin{tikzpicture}[node distance=1.5cm, semithick] \node[]
      (alice) {$A$}; \node[right of=alice, xshift=1.5cm] (bob) {$B$};
      \node[right of=alice, yshift=-0.5cm] (1cc)
      {$\onecc$};\node[right of=bob] (bc) {$\bc$}; \node[right of=bc]
      (env) {$\env$};

      \node[minimum height=1.4cm, draw, dashed, rounded corners,
      fit=(alice) (bob), label={210:$\siml$}] {};

      \path (alice) edge[bend left=15, <->] (bob) (bob) edge[bend
      left=10, <->] (1cc) (1cc) edge[bend left=10, <->] (alice);
      \draw (bob.east)+(0,0.1) edge[->] node[above] {\tiny(commit,
        $b$)} ($(bc.west)+(0,0.1)$);

      \path (bob.east)+(0,-0.1) edge[->] node[below] {\tiny open}
      ($(bc.west) +(0,-0.1)$)

      (env.west)+(0,0.1) edge[<-] node[above] {\tiny committed}
      ($(bc.east)+(0,0.1)$)

      (env.west)+(0,-0.1) edge[<-] node[below] {\tiny (open,$b$)}
      ($(bc.east)+(0,-0.1)$);

      \draw[<->, rounded corners] (alice.north) -| +(0,1) -| (env);
    \end{tikzpicture}\vspace{-3ex}
  \end{center}
  \caption{The standard way for constructing $\siml$: run dishonest
    Alice internally and simulate honest Bob and the calls to the
    functionality $\onecc$, and extract the input to $\bc$.}
  \label{fig:typicalsim}
\end{figure}

In all our proofs below, the honest party is simulated by $\siml$ by
running it {\em honestly}, up to possible small modifications that are
unnoticeable to the adversary, and that do not affect the (simulated)
honest party's output. As such, in our proofs below, for showing
indistinguishability of the real and the ideal model, it is sufficient
to argue that, in the ideal model, the output of the simulated honest
party equals what $\sf F$ outputs to $\env$ upon the input that is
provided by $\siml$.

\subsection{UC Security of $\ot$ from $\onecc$}

As explained in Section~\ref{sec:UC}, our protocol
$\textsc{bc}^\onecc$ does not seem to satisfy the UC security
definition in case of a corrupted verifier Bob. As such, we cannot
conclude UC security of the standard $\bc$-based $\ot$
scheme~\cite{bbcs91,c94} with $\bc$ instantiated by
$\textsc{bc}^{\onecc}$. Instead, we show UC security of $\ot$ from
$\onecc$ by means of the following strategy.

First, we show UC security of {\normalfont $\textsc{bc}^\onecc$}
against a corrupted committer Alice (Proposition~\ref{prop:bc}).
Then, we show that $\bc$ and $\onecc$ together imply $\twocc$
(actually, a variation of $\twocc$ that gives Alice the option to
abort) by means of a straightforward protocol
(Proposition~\ref{prop:2cc}), and we recall that $\twocc$ implies
$\ot$ by means of the protocol $\textsc{ot}^\twocc$
from~\cite{fkszz13}. Instantiating the underlying functionality $\bc$
by $\textsc{bc}^{\onecc}$ then gives us a protocol
$\textsc{ot}^\onecc$ with UC security against a corrupted receiver
(Lemma~\ref{lemma:otdr}).  Finally, it is rather straightforward to
prove UC security of $\textsc{ot}^\onecc$ against a corrupted sender
directly (Lemma~\ref{lemma:otds}).

\begin{proposition}\label{prop:bc}
  Protocol {\normalfont $\textsc{bc}^\onecc$} statistically quantum-UC-emulates
  {\normalfont $\bc$} against corrupted committer Alice.
\end{proposition}
\begin{proof}
  The construction of $\siml$ follows the paradigm outlined
  above. $\siml$ runs dishonest Alice internally, and it simulates
  honest Bob and $\onecc$ by running them honestly. Note that $\siml$
  gets to see Alice's inputs to $\onecc$.  Once Alice announces $g,w$
  and $s$ at the end of the commit phase, $\siml$ computes
  $b=g(\theta')\oplus w$, where $\theta'$ is the string of syndrome
  $s$ closest to the stored $\theta_{\overline t}$, and inputs
  ``(commit, $b$)'' into the $\bc$ functionality.  Finally, when
  corrupted Alice opens her commitment, $\siml$ inputs ``open'' into
  $\bc$ if Bob accepted the opening, and inputs ``abort'' if Bob
  aborted.

  It now follows immediately from Lemma~\ref{prop:NA1cc} that the bit
  $b'$ output by the simulated Bob equals the bit $b$ computed by
  $\siml$ and input to $\bc$, except with negligible probability. As
  such, real and ideal model are statistically indistinguishable.
\qed
\end{proof}

\begin{figure}[h]
  \begin{framed}
    \begin{center}
    \end{center}
    {\bf Parties}: The sender Alice and the receiver Bob.\\
    {\bf Inputs}: Alice receives $s_0,s_1\in \bool$ and Bob receives
    $c\in \bool$.
    \begin{enumerate}
    \item Alice inputs ``(commit, $s_0$)'' in $\bc$, Bob receives
      ``committed''.
    \item\label{item:CC} Alice and Bob invoke $\onecc$ with respective inputs $s_1$
      and $c$.
    \item If Alice receives $c=1$ from $\onecc$, she sends ``open'' to
      $\bc$. Bob receives $s_0$ from $\onecc$ and $s_1$ from $\bc$.
    \item Alice outputs $c$, Bob outputs $(s_0,s_1)$ if $c=1$ and
      $\bot$ if $c=0$. Bob outputs ``abort'' if $c=1$ but Alice
      refuses %or fails
      to open her commitment.
    \end{enumerate}
\end{framed}
\caption{Protocol $\textsc{2cc}^{\bc,\onecc}$.}
\label{fig:2cc}
\end{figure}

Consider the candidate $2$-bit cut-and-choose protocol
$\textsc{2cc}^{\bc,\onecc}$ from Fig.~\ref{fig:2cc}.  This protocol
does not implement the full-fledged $\twocc$ functionality, but a
variation $\twocc'$ that gives the sender the option to abort after it
sees the receiver's input $c$. This is because in the protocol the
sender can refuse to open its commitments (or try to cheat when
opening them so that the receiver will reject).  In that case, the
receiver will only learn one of the receiver's two inputs. This will
not influence the security of the resulting $\ot$ scheme since
aborting in any instance of $\twocc'$ will stop the protocol.

Formally, $\twocc'$ is described as follows: it first waits for inputs
$(s_0,s_1)$ from Alice and $c$ from Bob. Upon reception of both
inputs, it sends $c$ to Alice. If $c=0$, it sends $\bot$ to Bob. If
$c=1$, it waits for response ``abort'' or ``continue'' from Alice. On
input ``continue'', $\twocc'$ outputs $(s_0,s_1)$ to Bob and on input
``abort'', it outputs ``abort''.

\begin{proposition}\label{prop:2cc}
  Protocol {\normalfont $\textsc{2cc}^{\bc,\onecc}$} statistically
  quantum-UC-emulates {\normalfont $\twocc'$}. 
\end{proposition}
\begin{proof}
  We first consider a corrupted sender Alice.  $\siml$ simulates Bob,
  $\bc$ and $\onecc$ by running them honestly. After
  step~\ref{item:CC}, when $\siml$ has learned Alice's respective
  inputs $s_0$ and $s_1$ to $\bc$ and $\onecc$, it inputs $(s_0,s_1)$
  into the functionality $\twocc'$.  After receiving $c$ from the
  $\twocc'$, $\siml$ makes Bob input $c$ into the $\onecc$. If $c=0$
  then the simulated Bob and $\twocc'$ both output $\bot$. If $c=1$
  then Alice is supposed to open her commitment. If she refuses then
  $\siml$ inputs ``abort'' into $\twocc'$, and the simulated Bob and
  $\twocc'$ both output ``abort''. Otherwise, i.e., if Alice opens the
  commitment (to $s_0$), $\siml$ inputs ``continue'', and the
  simulated Bob and $\twocc'$ both output $(s_0,s_1)$.  This proves
  the claim for a corrupted sender Alice.  Security against a
  corrupted receiver Bob is similarly straightforward.  \qed
\end{proof}

\begin{corollary}\label{cor:2cc}
  Protocol {\normalfont $\textsc{2cc}^\onecc$}, obtained by replacing each instance
  of {\normalfont $\bc$} by {\normalfont $\textsc{bc}^\onecc$}, statistically quantum-UC-emulates
  {\normalfont $\twocc'$} against corrupted sender.
\end{corollary}
\begin{proof}
  Since $\textsc{bc}^\onecc$ statistically quantum UC-emulates $\bc$
  against malicious committer, and since the sender in
  $\textsc{2cc}^{\bc,\onecc}$ is the committer of $\bc$, we can
  replace $\bc$ with $\textsc{bc}^\onecc$ in protocol
  $\textsc{2cc}^{\bc,\onecc}$ and still maintain UC-security against
  corrupted sender.\qed
\end{proof}

\begin{figure}[h]
  \begin{framed}
    \begin{center}
    \end{center}
    {\bf Parameters}: A family $\mathcal F=\{\bool^n\rightarrow
    \bool^\ell\}$ of universal hash functions.\\
    {\bf Parties}: The sender Alice and the receiver Bob.\\
    {\bf Inputs}: Alice receives $s_0,s_1\in \bool^\ell$ and Bob
    receives $c\in \bool$.
    \begin{enumerate}
    \item Alice chooses $x^A\in_R\{0,1\}^n$ and $\theta^A\in_R
      \{+,\times\}^{n}$ and sends the state $\ket {x^A}_{\theta^A}$ to
      Bob.
    \item\label{st:measure} Upon reception, Bob chooses $\theta^B\in_R\{+,\times\}^n$
      and measures the received state in basis $ \theta^B$. Lets
      $x^B$ denote the measurement outcome. 
    \item For $i=1\dots n$, do \label{st:checking}
      \begin{enumerate}
      \item\label{step:commit} Alice and Bob perform protocol
        $\textsc{commit}^\onecc$ with Bob as the sender and input
        $\theta^B_i$.
      \item\label{step:CC} Alice chooses a selection bit $t_i\in_R\bool$ and they
        invoke an instance of $\onecc$ with Bob as the sender and
        inputs $t_i$ and $x^B_i$.
      \item\label{step:reveal} Whenever $t_i=1$, Bob opens the $i$th
        commitment using protocol $\textsc{reveal}^\onecc$.
      \end{enumerate}
    \item If for some $i$ s.t. $t_i=1$, $\theta_i^A=\theta_i^B$, but
      $x^B_i \neq x^A_i$, Alice aborts. Bob aborts if $t_i=1$ for more
      than $3n/5$ positions. Let $\hat x^A$ (resp. $\hat \theta^A,
      \hat x^B,\hat \theta^B$) be the restriction of $x^A$ (resp. $
      \theta^A, x^B, \theta^B$) to the indices $i$ for which $t_i=0$.
    \item\label{step:partition} Alice sends $\hat \theta^A$ to Bob. Bob constructs sets
      $I_c=\{i\mid \hat \theta^A_i = \hat \theta^B_i\}$ and
      $I_{1-c}=\{i\mid \hat \theta^A_i \neq \hat \theta^B_i\}$ then
      sends $(I_0, I_1)$ to Alice.
    \item Alice chooses $f\in_R \mathcal F$, computes $m_i=s_i\oplus
      f(\hat x^A_{I_i})$ for $i=0,1$ and sends $(f,m_0,m_1)$ to Bob.
    \item Bob outputs $s=m_c\oplus f(\hat x^B_{I_c})$.
    \end{enumerate}
\end{framed}
\caption{Protocol $\textsc{ot}^{\onecc}$.}
\label{fig:ot1cc}
\end{figure}

\begin{lemma}\label{lemma:otdr}
  Protocol {\normalfont $\textsc{ot}^\onecc$} statistically quantum UC-emulates
  {\normalfont $\ot$} for corrupted receiver.
\end{lemma}
\begin{proof}
  Note that steps~\ref{step:commit} through~\ref{step:reveal} of
  protocol $\textsc{ot}^\onecc$ are identical to protocol
  $\textsc{2cc}^\onecc$ defined above with Bob as the sender and Alice
  as the receiver. Since $\textsc{2cc}^\onecc$ statistically
  quantum-UC-emulates $\twocc'$ against corrupted sender, we may
  replace steps~\ref{step:commit}-\ref{step:reveal} by a single call
  to $\twocc'$ with Bob as the sender and Alice as the receiver, and
  analyze the security of this protocol instead. The only difference
  between this protocol and the $\twocc$-based oblivious-transfer
  protocol from~\cite{fkszz13} is that the former uses $\twocc'$
  instead. However, this change does not affect UC-security since any
  adversary that aborts during one of the $\textsc{2cc}^\onecc$
  subroutines is indistinguishable from an adversary that aborts right
  after the same subroutine. It directly follows from the analysis
  of~\cite{fkszz13}, that protocol $\textsc{ot}^\onecc$ statistically
  quantum-UC-emulates $\ot$ against corrupted receiver. \qed
\end{proof}

\begin{lemma}\label{lemma:otds}
  Protocol {\normalfont $\textsc{ot}^\onecc$} statistically quantum UC-emulates
  {\normalfont $\ot$} for corrupted sender.
\end{lemma}
\begin{proof}
  Let Alice be the corrupted sender and Bob the honest
  receiver. $\siml$ simulates Bob and $\onecc$ by running them
  honestly, {\em except} that Bob does {\em not} measure the received
  state in step~\ref{st:measure} but stores it, and in
  step~\ref{step:CC}, whenever Alice inputs $t_i=1$ into $\onecc$,
  $\siml$ ``rushes'' and measures the $i$th qubit in basis
  $\theta^B_i$ and inputs the outcome $x^B_i$ in the
  $\onecc$. Furthermore, in step~\ref{step:partition}, $\siml$ replies
  to Alice with a random partition $(I_0,I_1)$.  At the end of the
  protocol, $\siml$ measures the remaining qubits in Alice's basis
  $\hat \theta^A$ to obtain $\hat x^B$, computes $s_i=m_i\oplus f(\hat
  x^B_{I_i})$ for $i = 0,1$, and sends $(s_0, s_1)$ to the ideal $\ot$
  functionality.

  The output of $\ot$, i.e., $s_c$, coincides with the string that a
  fully honest Bob would have output; hence, we have
  indistinguishability between the real and the ideal model.
  \qed
\end{proof}

\begin{theorem}
  {\normalfont $\onecc$} is statistically quantum UC-complete.
\end{theorem}
\begin{proof}
  We have shown that $\textsc{ot}^\onecc$ statistically
  quantum-UC-emulates $\ot$. Since $\ot$ is quantum-UC-complete, we
  conclude that $\onecc$ is also quantum-UC-complete.\qed
\end{proof}

\bibliographystyle{plain}
\bibliography{big}

\end{document}